\documentclass[10pt]{article}
\usepackage{amssymb}
\usepackage{amsfonts}
\usepackage{bbm}
\usepackage{mathrsfs}
\usepackage{mathrsfs}
\usepackage{longtable}
\usepackage{supertabular}
\usepackage{amsfonts,mathrsfs,latexsym,amsmath,amssymb,amsthm}

\newtheorem{theorem}{Theorem}

\newtheorem{lemma}{Lemma}
\newtheorem{corollary}{Corollary}
\newtheorem{definition}{Definition}
\newtheorem{example}{Example}

\begin{document}
\title{\bf{On Linear Codes over $\mathbb{Z}_4+v\mathbb{Z}_4$}}
\author{{\bf  Jian Gao$^1$, Fang-Wei Fu$^1$, Yun Gao$^2$}\\
 {\footnotesize \emph{1. Chern Institute of Mathematics and LPMC, Nankai University}}\\
  {\footnotesize  \emph{Tianjin, 300071, P. R. China}}\\
  {\footnotesize \emph{2. School of Science, Shandong University of Techonology}}\\
  {\footnotesize  \emph{Zibo, 255091, P. R. China}}
  }
\date{}

\maketitle \noindent {\small {\bf Abstract} Linear codes are considered over the ring $\mathbb{Z}_4+v\mathbb{Z}_4$, where $v^2=v$. Gray weight, Gray maps for linear codes are defined and MacWilliams identity for the Gray weight enumerator is given. Self-dual codes, construction of Euclidean isodual codes, unimodular complex lattices, MDS codes and MGDS codes over $\mathbb{Z}_4+v\mathbb{Z}_4$ are studied. Cyclic codes and quadratic residue codes are also considered. Finally, some examples for illustrating the main work are given.}
 \vskip 1mm

\noindent
 {\small {\bf Keywords} MacWilliams identity; self-dual codes; unimodular complex lattices; cyclic codes; quadratic residue codes}

\vskip 3mm \noindent {\bf Mathematics Subject Classification (2000) } 11T71 $\cdot$ 94B05 $\cdot$ 94B15

\vskip 3mm \baselineskip 0.2in

\section{Introduction}
Error-Correcting codes play important roles in application ranging from date networking to satellite communication to compact disks. Most coding theory concerns on linear codes since they have clear structure that makes them easy to discover, to understand and to encode and decode.
\par
Codes over finite rings have been studied since the early 1970s. There are a lot of works on codes over finite rings after the discovery that certain good nonlinear binary codes can be constructed from cyclic codes over $\mathbb{Z}_4$ via the Gary map \cite{Hammons}. Since then, many researchers have payed more and more attentions to study the codes over finite rings. In these studies, the group rings associated with codes are finite chain rings.
\par
Recently, Zhu et al. considered linear codes over the finite non-chain ring $\mathbb{F}_q+v\mathbb{F}_q$. In \cite{Zhu1}, they study the cyclic codes over $\mathbb{F}_2+v\mathbb{F}_2$. It has shown that cyclic codes over this ring are principally generated. In the subsequent paper \cite{Zhu2}, they investigate a class of constacyclic codes over $\mathbb{F}_p+v\mathbb{F}_p$. In that paper, the authors prove that the image of a $(1-2v)$-constacyclic code of length $n$ over $\mathbb{F}_p+v\mathbb{F}_p$ under the Gray map is a cyclic code of length $2n$ over $\mathbb{F}_p$. Furthermore, they also assert that $(1-2v)$-constacyclic codes over $\mathbb{F}_p+v\mathbb{F}_p$ are also principally generated. More recently, Yildiz and Karadeniz \cite{Yildiz} studied the linear codes over the non-principal ring $\mathbb{Z}_4+u\mathbb{Z}_4$, where $u^2=0$. They introduce the MacWilliams identities for the complete, symmetrized and Lee weight enumerators. They also gave some methods to construct formally self-dual codes over $\mathbb{Z}_4+u\mathbb{Z}_4$.
\par
Self-dual codes are an important class of linear codes. They have connections to many fields of research such as lattices, designs and invariant \cite{Bannai,Cengellenmis}. The study of extremal self-dual codes and the connections with unimodular lattices has generated a lot of interests among the coding theory. And this is one of the motivations to introduce self-dual codes over the ring $\mathbb{Z}_4+v\mathbb{Z}_4$.
\par
As a special class of cyclic codes, quadratic residue codes fall into the family of BCH codes and have proven to be a promising family of cyclic codes. They were first introduced by Gleason and since then have generated a lot of interests. This due to the fact that they enjoy good algebra properties and they contain source of good codes. Recently, quadratic residue codes have been studied over some special rings \cite{Kaya,Pless2}.
\par
In this paper, we mainly introduce some results on linear codes over the principal ring $R=\mathbb{Z}_4+v\mathbb{Z}_4$, where $v^2=v$. To the best of our knowledge, this is the first time to study the linear codes over this ring. The remainder of this paper is organized as follows. In Section 2, we define the Gray weight of the element of $R$, and introduce a Gray map that leads to some useful results on linear codes over $R$. Moreover, we also give the MacWilliams identity on the linear code over $R$. In Section 3, we introduce some important class of linear codes: self-dual codes, MDS codes and MGDS codes. We give the sufficient and necessary conditions for a linear code to be Euclidean self-dual, MDS and MGDS codes. We also define the Hermitian dual on the linear code, and research the connections between Hermitian self-dual codes and unimodular complex lattices. Furthermore, we obtain a Gray distance bound on the code over $R$. In Section 4, we study the cyclic codes over $R$ including the generating polynomials, the generating idempotents and their duals. In Section5, we introduce an important class of cyclic codes called quadratic residue codes over $R$. Moreover, the extensions of quadratic residue codes are also discussed in this section. In Section 6, we give some examples to illustrate the main work in this paper.

\section{Linear codes over $R$}
Let $R=\mathbb{Z}_4+v\mathbb{Z}_4$, where $v^2=v$. Then $R$ is commutative and with characteristic $4$. Clearly, $R\simeq \mathbb{Z}_4[v]/(v^2-v)$. An element $r$ of $R$ can be expressed uniquely as $r=a+bv$, where $a,b\in \mathbb{Z}_4$. The ring $R$ has the following properties
\par
\par $\bullet$~There are $9$ different ideals of $R$, and they are $(1)$, $(v+1)$, $(v+2)$, $(v-1)$, $(2)$, $(v)$, $(2v-2)$, $(2v)$, $(0)$;
\par $\bullet$~$R$ is a principal ring;
\par $\bullet$~$(v+1)$ and $(v+2)$ are the maximal ideals of $R$;
\par $\bullet$~$R$ is not a finite chain ring.\\

Furthermore, for any element $r=a+bv$ of $R$, $r$ is a unit if and only if $a\not\equiv 0 ({\rm mod}2)$ and $a+b\not\equiv 0 ({\rm mod}2)$. Moreover, one can verify that if $r$ is a unit of $R$ then $r^2=1$.
\begin{definition}
Let $r=a+bv$ be any element of $R$. Then the Gray weight of the element $r$ is defined as
   \begin{equation*}
   w_G(r)=w_L(a)+w_L(a+b),
   \end{equation*}
where the symbol $w_L(\Box)$ denotes the Lee weight of the element $\Box$ of $\mathbb{Z}_4$.
\end{definition}
\par
Define a Gray weight of a vector $\textbf{c}=(c_0, c_1, \ldots, c_{n-1})\in R^n$ to be the rational sum of the Gray weight of its components, i.e. $w_G(\textbf{c})=\sum_{i=0}^{n-1}w_G(c_i)$. For any elements $\textbf{c}_1, \textbf{c}_2 \in R^n$, the Gray distance is given by $d_G(\textbf{c}_1, \textbf{c}_2)=w_G(\textbf{c}_1- \textbf{c}_2)$. A code $\mathcal {C}$ of length $n$ over $R$ is a subset of $R^n$. $\mathcal {C}$ is linear if and only if $\mathcal {C}$ is an $R$-submodule of $R^n$. The minimum Gray distance of $\mathcal {C}$ is the smallest nonzero Gray distance between all pairs of distinct codewords. The minimum Gray weight of $\mathcal {C}$ is the smallest nonzero Gray weight among all codewords. If $\mathcal {C}$ is linear, then the minimum Gray distance is the same as the minimum Gray weight.
\par
Now we give the definition of the Gray map on $R^n$ as follows
\begin{equation*}
\begin{split}
\Phi:~~R^n&\rightarrow \mathbb{Z}_4^{2n}\\
 (c_0, c_1, \ldots, c_{n-1})&\mapsto (a_0, a_0+b_0, a_1, a_1+b_1, \ldots, a_{n-1}, a_{n-1}+b_{n-1}),
 \end{split}
 \end{equation*}
where $c_i=a_i+b_iv$, $i=0,1,\ldots,n-1$.
\par
It is well known that the Lee weights of the elements in $\mathbb{Z}_4$ are defined as $w_L(0)=0, w_L(1)=w_L(3)=1$ and $w_L(2)=2$. Then we have the following result.
\begin{theorem}
The Gray map $\Phi$ is a distance-preserving map from $R^n$ (Gray distance) to $\mathbb{Z}_4^{2n}$ (Lee distance) and it is also $\mathbb{Z}_4$-linear.
\end{theorem}
\begin{proof}
Let $k_1, k_2 \in \mathbb{Z}_4$. Then, by the definition of Gray map $\Phi$, for any $\textbf{c}_1, \textbf{c}_2\in R^n$ we have $\Phi(k_1\textbf{c}_1+k_2\textbf{c}_2)=k_1\Phi(\textbf{c}_1)+k_2\Phi(\textbf{c}_2)$, which implies that $\Phi$ is $\mathbb{Z}_4$-linear. Let $\textbf{c}_1=(c_{1,0}, c_{1,1}, \ldots, c_{1,n-1})$ and $\textbf{c}_2=(c_{2.0}, c_{2,1}, \ldots, c_{2,n-1})$ be elements of $R^n$, where $c_{1,i}=a_{1,i}+b_{1,i}v$ and $c_{2,i}=a_{2,i}+b_{2,i}v$, $i=0,1,\ldots,n-1$. Then $\textbf{c}_1-\textbf{c}_2=(c_{1,0}-c_{2,0}, \ldots, c_{1,n-1}-c_{2,n-1})$ and $\Phi(\textbf{c}_1-\textbf{c}_2)=\Phi(\textbf{c}_1)-\Phi(\textbf{c}_2)$. Therefore $d_G(\textbf{c}_1, \textbf{c}_2)=w_G(\textbf{c}_1- \textbf{c}_2)=w_L(\Phi(\textbf{c}_1-\textbf{c}_2))=w_L(\Phi(\textbf{c}_1)-\Phi(\textbf{c}_2))=d_L(\Phi(\textbf{c}_1), \Phi(\textbf{c}_2))$. The second equality above holds because of the definition of the Gray weight of the element in $R$.
\end{proof}
\begin{lemma}
Let $\mathcal {C}$ be a $(n, M, d)$ linear code over $R$, where $n, M, d$ are the length, the number of the codewords and the minimum Gray distance of $\mathcal {C}$,respectively. Then $\Phi(\mathcal {C})$ is a $(2n, M, d)$ linear code over $\mathbb{Z}_4$.
\end{lemma}
\begin{proof}
From Theorem 1, we see that $\Phi(\mathcal {C})$ is $\mathbb{Z}_4$-linear, which implies that $\Phi(\mathcal {C})$ is a $\mathbb{Z}_4$-linear code. From the definition of the Gray map $\Phi$, $\Phi(\mathcal {C})$ is with length $2n$. Moreover, one can check that $\Phi$ is a bijective map from $R^n$ to $\mathbb{Z}_4^{2n}$ implying that $\Phi(\mathcal {C})$ has $M$ codewords. At last, the preserving distance of $\Phi$ leads to $\Phi(\mathcal {C})$ has the minimum Lee distance $d$.
\end{proof}

Let $\textbf{x}=(x_1, x_2, \ldots, x_n)$ and $\textbf{y}=(y_1, y_2, \ldots, y_n)$ be two vectors of $R^n$. The Euclidean inner product of $\textbf{x}$ and $\textbf{y}$ is defined as follows
\begin{equation*}
\textbf{x}\cdot \textbf{y}=\sum_{i=1}^nx_iy_i.
\end{equation*}
The Euclidean dual code $\mathcal {C}^\perp$ of $\mathcal {C}$ is defined as $\mathcal {C}^\perp=\{\textbf{x}\in R^n| \textbf{x}\cdot \textbf{c}=0~{\rm for ~all~}\textbf{c}\in \mathcal {C}\}$. $\mathcal {C}$ is said to be Euclidean self-orthogonal if $\mathcal {C}\subseteq \mathcal {C}^\perp$ and Euclidean self-dual if $\mathcal {C}=\mathcal {C}^\perp$.
\begin{theorem}
Let $\mathcal {C}$ be a linear code. Then $\Phi(\mathcal {C})^\perp=\Phi(\mathcal {C}^\perp)$. Moreover, if $\mathcal {C}$ is Euclidean self-dual, so is $\Phi(\mathcal {C})$.
\end{theorem}
\begin{proof}
For all $\textbf{c}_1=(c_{1,0}, c_{1,1}, \ldots, c_{1,n-1})\in \mathcal {C}$ and $\textbf{c}_2=(c_{2,0}, c_{2,1}, \ldots, c_{2,n-1})\in \mathcal {C}^\perp$, where $c_{j,i}=a_{j,i}+b_{j,i}v$, $a_{j,i}, b_{j,i}\in \mathbb{Z}_4$, $j=1,2$, $i=0,1,\ldots,n-1$, if $\textbf{c}_1\cdot \textbf{c}_2=0$, then we have $\textbf{c}_1\cdot \textbf{c}_2=\sum_{i=0}^{n-1}c_{1,i}c_{2,i}=\sum_{i=0}^{n-1}a_{1,i}a_{2,i}+\sum_{i=0}^{n-1}(a_{1,i}b_{2,i}+a_{2,i}b_{1,i}+b_{1,i}b_{2,i})v=0$ implying $\sum_{i=0}^{n-1}a_{1,i}a_{2,i}=0$ and $\sum_{i=0}^{n-1}(a_{1,i}b_{2,i}+a_{2,i}b_{1,i}+b_{1,i}b_{2,i})=0$. Therefore, $\Phi(\textbf{c}_1)\cdot \Phi(\textbf{c}_2)=\sum_{i=0}^{n-1}(a_{1,i}a_{2,i}+a_{1,i}b_{2,i}+a_{2,i}b_{1,i}+b_{1,i}b_{2,i})=0$. Thus $\Phi(\mathcal {C}^\perp)\subseteq \Phi(\mathcal {C})^\perp$. From Lemma 1, we can verify that $|\Phi(\mathcal {C}^\perp)|=|\Phi(\mathcal {C})^\perp|$, which implies that $\Phi(\mathcal {C})^\perp=\Phi(\mathcal {C}^\perp)$. Clearly, $\Phi(\mathcal {C})$ is Euclidean self-orthogonal if $\mathcal {C}$ is Euclidean self-dual. From Lemma 1, we have $|\Phi(\mathcal {C})|=|\mathcal {C}|=16^{n/2}=4^{2n/2}$. Thus, $\Phi(\mathcal {C})$ is Euclidean self-dual.
\end{proof}

One of the most remarkable results on coding theory is the MacWilliams identity that describes the connections between a linear code and its dual code on the weight enumerator. In the following, we discuss this issue over $R$.
\par
Let $\mathcal {C}$ be a linear code of length $n$ over $R$. Suppose that $a$ is any element of $R$. For all $\textbf{c}=(c_0, c_1, \ldots, c_{n-1})\in R^n$, define the weight of $\textbf{c}$ at $a$ to be $w_a=|\{i|c_i=a\}|$.
\begin{definition}
Let $A_i$ be the number of elements of the Gray weight $i$ in $\mathcal {C}$. Then the set $\{A_0, A_1, \ldots, A_{4n}\}$ is called the Gray weight distribution of $\mathcal {C}$. Define the Gray weight enumerator of $\mathcal {C}$ as $Gray_\mathcal {C}(X, Y)=\sum_{i=0}^{4n}A_iX^{4n-i}Y^i$. Clearly, $Gray_\mathcal {C}(X, Y)=\sum_{\textbf{c}\in \mathcal {C}}X^{4n-w_G(\textbf{c})}Y^{w_G(\textbf{c})}$.
   Furthermore, define the complete weight enumerator of $\mathcal {C}$ as $cwe_\mathcal {C}(X_0, X_1, X_2, X_3, X_v, X_{1+v}, X_{2+v}, X_{3+v}, X_{2v}, X_{1+2v}, X_{2+2v}, X_{3+2v}, \\X_{3v}, X_{1+3v}, X_{2+3v}, X_{3+3v})=\sum_{\textbf{c}\in \mathcal {C}}X_0^{w_0(\textbf{c})}X_1^{w_1(\textbf{c})}X_2^{w_2(\textbf{c})}X_3^{w_3(\textbf{c})}X_v^{w_v(\textbf{c})}X_{1+v}^{w_{1+v}(\textbf{c})}
   X_{2+v}^{w_{2+v}(\textbf{c})}X_{3+v}^{w_{3+v}(\textbf{c})}X_{2v}^{w_{2v}(\textbf{c})}\\
   X_{1+2v}^{w_{1+2v}(\textbf{c})}X_{2+2v}^{w_{2+2v}(\textbf{c})}X_{3+2v}^{w_{3+2v}(\textbf{c})}X_{3v}^{w_{3v}(\textbf{c})}X_{1+3v}^{w_{1+3v}(\textbf{c})}
   X_{2+3v}^{w_{2+3v}(\textbf{c})}X_{3+3v}^{w_{3+3v}(\textbf{c})}$.
\end{definition}

For any codeword $\textbf{c}$ of $\mathcal {C}$, let
\begin{equation*}
\alpha_0(\textbf{c})=w_0(\textbf{c})
\end{equation*}
\begin{equation*}
\alpha_1(\textbf{c})=w_v(\textbf{c})+w_{3+v}(\textbf{c})+w_{1+3v}(\textbf{c})
\end{equation*}
\begin{equation*}
\alpha_2(\textbf{c})=w_1(\textbf{c})+w_3(\textbf{c})+w_{2v}(\textbf{c})+w_{3v}(\textbf{c})+w_{1+2v}(\textbf{c})+w_{2+2v}(\textbf{c})+w_{3+2v}(\textbf{c})
\end{equation*}
\begin{equation*}
\alpha_3(\textbf{c})=w_{1+v}(\textbf{c})+w_{2+v}(\textbf{c})+w_{2+3v}(\textbf{c})+w_{3+3v}(\textbf{c})
\end{equation*}
\begin{equation*}
\alpha_4(\textbf{c})=w_2(\textbf{c})
\end{equation*}
denote the number of elements of $\textbf{c}$ with Gray weight $0,1,2,3,4$, respectively. Then the Gray weight $w_G(\textbf{c})$ of $\textbf{c}\in \mathcal {C}$ is defined to be $$w_G(\textbf{c})=\alpha_1(\textbf{c})+2\alpha_2(\textbf{c})+3\alpha_3(\textbf{c})+4\alpha_4(\textbf{c}).$$
Define the symmetrized weight enumerator of $\mathcal {C}$ as $swe_\mathcal {C}(X_0, X_1, X_2, X_3, X_4)=cwe_\mathcal {C}(X_0, X_1, X_2, X_3, X_v, X_{1+v}, \\ X_{2+v},X_{3+v}, X_{2v}, X_{1+2v}, X_{2+2v}, X_{3+2v}, X_{3v}, X_{1+3v}, X_{2+3v}, X_{3+3v})=\sum_{\textbf{c}\in \mathcal {C}} X_0^{\alpha_0(\textbf{c})}X_1^{\alpha_1(\textbf{c})}X_2^{\alpha_2(\textbf{c})}X_3^{\alpha_3(\textbf{c})}X_4^{\alpha_4(\textbf{c})}$. Furthermore, the Hamming weight enumerator of $\mathcal {C}$ is defined as
\begin{equation*}
Ham_\mathcal {C}(X, Y)=\sum_{\textbf{c}\in \mathcal {C}}X^{n-w_H(\textbf{c})}Y^{w_H(\textbf{c})},
\end{equation*}
where $w_H(\textbf{c})$ denotes the Hamming weight of the codeword $\textbf{c}$. Then we have the following results.
\begin{theorem}
Let $\mathcal {C}$ be a linear code of length $n$ over $R$. Then\\
{\rm (i)}~$Gray_\mathcal {C}(X, Y)=swe_\mathcal {C}(X^4, X^3Y, X^2Y^2, XY^3, Y^4)$;\\
{\rm (ii)}~$Ham_\mathcal {C}(X, Y)=swe_\mathcal {C}(X, Y, Y, Y, Y)$;\\
{\rm (iii)}~$Gray_\mathcal {C}(X, Y)=Lee_{\Phi(\mathcal {C})}(X, Y)$;\\
{\rm (iv)}~$Gray_{\mathcal {C}^\perp}(X, Y)=\frac{1}{|\mathcal {C}|}Gray_\mathcal {C}(X+Y, X-Y)$.
\end{theorem}
\begin{proof}
(i)~From the definition of the symmetrized weight enumerator, we have
\begin{equation*}
\begin{split}
swe_\mathcal {C}(X^4, X^3Y,&  X^2Y^2, XY^3, Y^4)\\
                           &=\sum_{\textbf{c}\in\mathcal {C}}X^{4\alpha_0(\textbf{c})}(X^3Y)^{\alpha_1(\textbf{c})}(X^2Y^2)^{\alpha_2(\textbf{c})}(XY^3)^{\alpha_3(\textbf{c})}Y^{4\alpha_4(\textbf{c})}\\
                           &=\sum_{\textbf{c}\in\mathcal {C}}X^{4\alpha_0(\textbf{c})+3\alpha_1(\textbf{c})+2\alpha_2(\textbf{c})+\alpha_3(\textbf{c})}Y^{\alpha_1(\textbf{c})+2\alpha_2(\textbf{c})
                           +3\alpha_3(\textbf{c})+4\alpha_4(\textbf{c})}\\
                           &=\sum_{\textbf{c}\in\mathcal {C}}X^{4n-w_G(\textbf{c})}Y^{w_G(\textbf{c})}\\
                           &=Gray_\mathcal {C}(X, Y).
 \end{split}
\end{equation*}
(ii)~From the definition of symmetrized weight enumerator, we have
\begin{equation*}
\begin{split}
swe_\mathcal {C}(X, Y, Y, Y, Y)&=\sum_{\textbf{c}\in\mathcal {C}} X^{\alpha_0(\textbf{c})}Y^{\alpha_1(\textbf{c})}Y^{\alpha_2(\textbf{c})}Y^{\alpha_3(\textbf{c})}Y^{\alpha_4(\textbf{c})}\\
                               &=\sum_{\textbf{c}\in\mathcal {C}}X^{\alpha_0(\textbf{c})}Y^{\alpha_1(\textbf{c})+\alpha_2(\textbf{c})+\alpha_3(\textbf{c})+\alpha_4(\textbf{c})}\\
                               &=\sum_{\textbf{c}\in\mathcal {C}}X^{n-w_H(\textbf{c})}Y^{w_H(\textbf{c})}\\
                               &=Ham_\mathcal {C}(X, Y).
 \end{split}
 \end{equation*}
(iii)~From the definition of Gray weight enumerator, we obtain that
\begin{equation*}
\begin{split}
Gray_\mathcal {C}(X, Y)&=\sum_{\textbf{c}\in\mathcal {C}}X^{4n-w_G(\textbf{c})}Y^{w_G(\textbf{c})}\\
                       &=\sum_{\Phi(\textbf{c})\in \Phi(\mathcal {C})}X^{4n-w_L(\Phi(\textbf{c}))}Y^{w_L(\Phi(\textbf{c}))}\\
                       &=Lee_{\Phi(\mathcal {C})}(X, Y).
 \end{split}
\end{equation*}
(iv)~From Theorem 2, $\Phi(\mathcal {C}^\perp)=\Phi(\mathcal {C})^\perp$ and they are both $\mathbb{Z}_4$-linear according to Lemma 1. By Theorem 2.4 in \cite{Wan} and (iii), we have
\begin{equation*}
\begin{split}
Gray_{\mathcal {C}^\perp}(X, Y)&=Lee_{\Phi(\mathcal {C}^\perp)} (X, Y)\\
                               &=Lee_{\Phi(\mathcal {C})^\perp}(X, Y)\\
                               &=\frac{1}{|\Phi(\mathcal {C})|}Lee_{\Phi(\mathcal {C})}(X+Y, X-Y)\\
                               &=\frac{1}{|\mathcal {C}|}Gray_{\mathcal {C}}(X+Y, X-Y).
 \end{split}
 \end{equation*}
\end{proof}
\section{Self-dual codes, MDS codes and MGDS codes}
Self-dual codes, MDS codes and MGDS codes are important classes of linear codes. They have been studied over a wide variety of rings, including finite fields, Galois rings and finite chain rings. In this section, we investigate some properties of these codes over $R$.
\subsection{Euclidean Self-dual codes over $R$}
Euclidean self-dual codes over rings have been shown to have closely interesting connections to the invariant theory, lattice theory and the theory of modular forms. At the beginning, we introduce some useful facts.
\par By the Chinese Remainder Theorem, we have
\begin{equation*}
\begin{split}
R&=vR\oplus (1-v)R \\
 &=v\mathbb{Z}_4\oplus (1-v)\mathbb{Z}_4.
 \end{split}
 \end{equation*}
  Define $$\mathcal {C}_1=\{\textbf{x}\in \mathbb{Z}_4^n| \exists \textbf{y}\in \mathbb{Z}_4^n, v\textbf{x}+(1-v)\textbf{y}\in \mathcal {C}\}$$ and $$\mathcal {C}_2=\{\textbf{y}\in \mathbb{Z}_4^n| \exists \textbf{x}\in \mathbb{Z}_4^n, v\textbf{x}+(1-v)\textbf{y}\in \mathcal {C}\}.$$ Then $\mathcal {C}_1$ and $\mathcal {C}_2$ are both $\mathbb{Z}_4$-linear of length $n$. Moreover, the linear code $\mathcal {C}$ of length $n$ over $R$ can be uniquely expressed as $$\mathcal {C}=v\mathcal {C}_1\oplus (1-v)\mathcal {C}_2.$$
\begin{theorem}
Let $\mathcal {C}$ be a linear code of length $n$ over $R$. Then $\mathcal {C}^\perp=v\mathcal {C}_1^\perp\oplus(1-v)\mathcal {C}_2^\perp$. Moreover, $\mathcal {C}$ is Euclidean self-dual if and only if $\mathcal {C}_1$ and $\mathcal {C}_2$ are both Euclidean self-dual over $\mathbb{Z}_4$.
\end{theorem}
\begin{proof}
Define
\begin{equation*}
\widehat{\mathcal {C}}_1=\{ \textbf{x} \in \mathbb{Z}_4^n|~\exists \textbf{y} \in \mathbb{Z}_4^n, v\textbf{x}+(1-v)\textbf{y}\in \mathcal {C}^\perp\}
\end{equation*}
and
\begin{equation*}
\widehat{\mathcal {C}}_2=\{ \textbf{y} \in \mathbb{Z}_4^n|~\exists \textbf{x} \in \mathbb{Z}_4^n, v\textbf{x}+(1-v)\textbf{y}\in \mathcal {C}^\perp\}.
\end{equation*}
Then $\mathcal {C}^\perp=v\widehat{\mathcal {C}}_1+(1-v)\widehat{\mathcal {C}}_2$ and this expression is unique. Clearly, $\widehat{\mathcal {C}}_1\subseteq \mathcal {C}_1^\perp$. Let $\textbf{c}_1$ be an element of $\mathcal {C}_1^\perp$. Then, for any $\textbf{x}\in \mathcal {C}_1$, there exists $\textbf{y}\in \mathbb{Z}_4^n$ such that $\textbf{c}_1\cdot(v\textbf{x}+(1-v)\textbf{y})=\textbf{0}$. Let $\textbf{c}=v\textbf{x}+(1-v)\textbf{y}\in \mathcal {C}$. Then $v\textbf{c}_1\cdot \textbf{c}=\textbf{0}$, which implies that $v\textbf{c}_1\in \mathcal {C}^\perp$. By the unique expression of $\mathcal {C}^\perp$, we have $\textbf{c}_1\in \widehat{\mathcal {C}}_1$, i.e. $\mathcal {C}_1=\widehat{\mathcal {C}}_1$. Similarly, we can prove $\mathcal {C}_2=\widehat{\mathcal {C}}_2$ implying $\mathcal {C}^\perp=v\mathcal {C}_1^\perp+(1-v)\mathcal {C}_2^\perp$.
\par
Clearly, $\mathcal {C}$ is Euclidean self-dual over $R$ if $\mathcal {C}_1$ and $\mathcal {C}_2$ are both Euclidean self-dual over $\mathbb{Z}_4$. If $\mathcal {C}$ is Euclidean self-dual, then $\mathcal {C}_1$ and $\mathcal {C}_2$ are both Euclidean self-orthogonal over $\mathbb{Z}_4$, i.e. $\mathcal {C}_1\subseteq \mathcal {C}_1^\perp$ and $\mathcal {C}_2\subseteq \mathcal {C}_2^\perp$. Next, we will prove $\mathcal {C}_1=\mathcal {C}_1^\perp$ and  $\mathcal {C}_2=\mathcal {C}_2^\perp$. If not, then there are elements $\textbf{a}\in \mathcal {C}_1^\perp \setminus \mathcal {C}_1$ and $\textbf{b}\in \mathcal {C}_2$ such that $(v\textbf{a}+(1-v)\textbf{b})^2\neq \textbf{0}$, which is a contradiction that $\mathcal {C}$ is Euclidean self-dual. Therefore, $\mathcal {C}_1=\mathcal {C}_1^\perp$ and $\mathcal {C}_2=\mathcal {C}_2^\perp$.
\end{proof}

For Euclidean self-dual codes, the conditions of existing are very important for the enumeration.
\begin{theorem}
There exist Euclidean self-dual codes of any length $n$ over $R$.
\end{theorem}
\begin{proof}
Firstly, the element $2$ of $R$ generates a Euclidean self-dual code of length $1$ over $R$. Secondly, we assert that if $\mathcal {C}$ and $\mathcal {D}$ are both Euclidean self-dual codes of length $n$ and $m$ over $R$ respectively, then the direct product $\mathcal {C}\times \mathcal {D}$ is also a Euclidean self-dual code of length $n+m$ over $R$. In fact, let $(\textbf{c}_1 ,\textbf{d}_1), (\textbf{c}_2, \textbf{d}_2)\in \mathcal {C}\times \mathcal {D}$. Then $(\textbf{c}_1 ,\textbf{d}_1)\cdot(\textbf{c}_2, \textbf{d}_2)=(\textbf{c}_1\cdot \textbf{c}_2, \textbf{d}_1\cdot\textbf{d}_2)=(\textbf{0}, \textbf{0})$, which implies that $\mathcal {C}\times \mathcal {D}$ is Euclidean self-orthogonal. Moreover, since $\mathcal {C}$ and $\mathcal {D}$ are both Euclidean self-dual over $R$, it follows that $|\mathcal {C}|=|R|^{n/2}$ and $|\mathcal {D}|=|R|^{m/2}$. Therefore $|\mathcal {C}\times \mathcal {D}|=|\mathcal {C}||\mathcal {D}|=|R|^{(n+m)/2}$ implying $\mathcal {C}\times \mathcal {D}$ is Euclidean self-dual.
\end{proof}

For a $\mathbb{Z}_4$-linear code $C$, $C$ and its Euclidean dual $C^\perp$ have $G$ and $G^\perp$ as their standard generator matrices, respectively
\begin{equation*}
G=\left(
\begin{array}{ccc}
I_{k_1} & A & B \\
\textbf{0} & 2I_{k_2} & 2 C\\
\end{array}
\right),
\end{equation*}
\begin{equation*}
G^\perp=\left(
\begin{array}{ccc}
-B^t-C^tA^t & C^t & I_{n-k_1-k_2} \\
2A^t        & 2I_{k_2}  & \textbf{0}  \\
\end{array}
\right).
\end{equation*}

Furthermore, $C$ and $C^\perp$ are of the type $4^{k_1}2^{k_2}$ and $4^{n-k_1-k_2}2^{k_2}$, respectively. Therefore $C$ is Euclidean self-dual over $\mathbb{Z}_4$ if and only if $C$ and $C^\perp$ are of the same type, which implies that $C$ is of type $4^k2^{n-2k}$. Then, by Theorem 4, Theorem 5 and Theorem 12.5.7 \cite{Huffuman}, we have the following straightforward result.
\begin{theorem}
For $0\leq k\leq \lfloor n/2 \rfloor $, the total number of Euclidean self-dual code over $R$ of length $n$ is
\begin{equation*}
(\sum_{k=0}^{\lfloor n/2\rfloor }\nu_{n, k}2^{k(k+1)/2})^2,
\end{equation*}
where $\nu_{n, k}$ is the number of $[n, k]$ Euclidean self-orthogonal doubly-even (i.e. the Hamming weight of every codeword is divisible by $4$ ) binary codes.
\end{theorem}

In the following of this section, we discuss some special class of Euclidean self-dual codes over $R$. It needs the following definition first.
\begin{definition}
Let $r=a+vb$ be an element of $R$. Then the Euclidean weight of $r$ is defined as follows
\begin{equation*}
w_E(r)=w_E(a)+w_E(a+b),
\end{equation*}
where $$w_E(a)={\rm min}\{ |a|^2, |4-a|^2\}$$ and $$w_E(a+b)={\rm min}\{|a+b|^2, |4-a-b|^2\}.$$
The Euclidean weight of a vector $\textbf{c}=(c_0, c_1, \ldots, c_{n-1})\in R^n$ is the rational sum of the Euclidean weight of its components, i.e. $w_E(\textbf{c})=\sum_{i=0}^{n-1}w_E(c_i)$.
\end{definition}
\begin{lemma}
The Gray map $\Phi$ is Euclidean weight-preserving from $R^n$ to $\mathbb{Z}_4^{2n}$.
\end{lemma}
\begin{proof}
It is well known that the Euclidean weight of the element $a$ of $\mathbb{Z}_4$ is defined as $w_E(a)={\rm min}\{ |a|^2, |4-a|^2\}$ and the Euclidean weight of a vector $\textbf{c}=(c_0, c_1, \ldots, c_{n-1})\in \mathbb{Z}_4^n$ is the rational sum of the Euclidean weight of its components, i.e. $w_E(\textbf{c})=\sum_{i=0}^{n-1}w_E(c_i)$. Then, by the definitions of the Gray map $\Phi$ and the Euclidean weight of the element of $R$, we can show that $\Phi$ is a Euclidean weight-preserving map from $R^n$ to $\mathbb{Z}_4^{2n}$.
\end{proof}

A Euclidean self-dual code $\mathcal {C}$ of length $n$ over $R$ is called Type II if the Euclidean weight of every codeword of $\mathcal {C}$ is multiple of $8$, otherwise $\mathcal {C}$ is called Type I.
\begin{theorem}
Let $\mathcal {C}$ be a Euclidean self-dual code of length $n$ over $R$. Then\\
{\rm (i)}~$\mathcal {C}$ is Type II if and only if $n$ is multiple of $4$.\\
{\rm (ii)}~If $\mathcal {C}$ is Type II, so is $\Phi(\mathcal {C})$.\\
{\rm (iii)}~The minimum Euclidean weight of $\mathcal {C}$ satisfies $$d_E\leq8\lfloor n/12\rfloor+8.$$
\end{theorem}
\begin{proof}
Let $\mathcal {C}$ be a Euclidean self-dual code of length $n$ over $R$. Then, by Theorem 2, $\Phi(\mathcal {C})$ is a Euclidean self-dual code of length $2n$ over $\mathbb{Z}_4$. From Lemma 2, we have that (ii) is valid. For (i), it is well known that there exist self-dual codes of length $n$ over $\mathbb{Z}_4$ if and only if $n$ is multiple of $8$ \cite{Bannai}, which follows (i). (iii) is follows from the Theorem 12.5.1 of \cite{Huffuman}.
\end{proof}

The Euclidean self-dual codes meeting the bound in Theorem 7(iii) are called Euclidean extremal. By Lemma 2 and Theorem 7(ii), if $\mathcal {C}$ is a Euclidean extremal Type II code over $R$, so is $\Phi(\mathcal {C})$ over $\mathbb{Z}_4$. The bound of Theorem 7(iii) is obviously the bound on the minimum Gray weight of Euclidean self-dual codes over $R$, but highly unsatisfactory one. An useful further studying is how to construct Type II codes over $R$.
\par Similarly, by the Theorem 12.5.8 in \cite{Huffuman}, we have the following result on the enumerator of Type II codes over $R$.
\begin{theorem}
Let $n\equiv 0({\rm mod}~4)$ and $N=2n$. Then the total number of Type II codes of length $n$ over $R$ is
\begin{equation*}
\sum_{i=0}^{n}\mu_{N, k}2^{1+k(k-1)/2},
\end{equation*}
where $\mu_{N, k}$ is the number of $[N, k]$ Euclidean self-orthogonal doubly-even binary codes containing the vector $\textbf{1}$.
\end{theorem}

A Euclidean self-dual code has the obvious property that its weight distribution is the same as that of its dual. Some of the results proved for Euclidean self-dual codes require only this equality of weight distributions. This has led to a broader class of codes known as Euclidean formally self-dual codes that include Euclidean self-dual codes. A code $\mathcal {C}$ is called \emph{Euclidean formally self-dual} provided $\mathcal {C}$ and $\mathcal {C}^\perp$ have the same weight distribution. A subfamily of formally Euclidean self-dual codes is the class of Euclidean isodual codes. A code $\mathcal {C}$ is called \emph{Euclidean isodual} if it is equivalent to its dual. Clearly, any Euclidean isodual code is Euclidean formally self-dual; however, a Euclidean formally self-dual code need not to be Euclidean isodual. In the rest of this subsection, we will extend three methods described in \cite{Huffuman} to construct Euclidean isodual codes over $R$.
\vskip 3mm \noindent
{\bf Construction A.}\emph{ Let $M$ be an $n\times n$ matrix over $R$ such that $M^T=M$. Then the code generated by the matrix $[I_n|M]$ is a Euclidean isodual code of length $2n$ over $R$.}
\begin{proof}
Consider the matrix $G=[-M^T|I_n]=[-M|I_n]$ and $\mathcal {G}=[I_n|M]$. Let $\mathscr{C}$ and $\mathcal {C}$ be the codes generated by $G$ and $\mathcal {G}$, respectively. Clearly, $\mathcal {C}$ and $\mathscr{C}$ are equivalent and with size $16^n$. Then we just need to show $\mathscr{C}=\mathcal {C}^\perp$.
\par
Let $\textbf{v}$ be the $i$-th row of $G$ and $\textbf{w}$ be the $j$-th row of $\mathcal {G}$, respectively. Then $\textbf{v}\cdot \textbf{w}=0$ since $M^T=M$, which implies that $\mathscr{C}=\mathcal {C}^\perp$ and $\mathcal {C}$ is equivalent to its dual $\mathcal {C}^\perp$.
\end{proof}
\vskip 3mm \noindent
{\bf Construction B.} \emph{Let $M$ be an $n\times n$ circulant matrix over $R$. Then the matrix $[I_n|M]$ generates a Euclidean isodual code of length $2n$ over $R$.}
\begin{proof}
Let $\mathcal {C}$ and $\mathscr{C}$ be the codes generated by $[I_n|M]=G$ and $[-M^T|I_n]=\mathcal {G}$, respectively. Then $\mathscr{C}=\mathcal {C}^\perp$. Obviously, $\mathscr{C}$ is equivalent to the code $\mathcal {D}$ generated by $[M^T|I_n]$. Therefore, we just need to show that $\mathcal {C}$ is equivalent to $\mathcal {D}$ as follows.

(i)~Apply a row permutation $\sigma$ such that the first column of $\sigma(M^T)$ is the same as the first column of $M$. Since $M$ is circulant, it follows that every column of $M$ is then equal to a column of $\sigma(M^T)$.

(ii)~Apply a column permutation $\tau$ such that $\tau(\sigma(M^T))=M$.

(iii)~Use another column permutation $\rho$ such that $\rho(\sigma(I_n))=I_n$.

Then we obtain the matrix $G$ from $[M^T|I_n]$ by the consecutive applications of $\sigma,\tau,\rho$, which implies that $\mathcal {C}$ is equivalent to $\mathcal {D}$. Therefore $\mathcal {C}$ is equivalent to $\mathscr{C}=\mathcal {C}^\perp$.
\end{proof}
\vskip 3mm \noindent
{\bf Construction C.} \emph{Let $B$ be an $n\times n$ bordered circulant matrix as follows
$$B=\left[\begin{array}{cccc}\alpha&\beta&\cdots&\beta\\ \gamma& & &\\\vdots& & \huge{M} &\\\gamma& & & \end{array}\right],$$
where $M$ is an $(n-1)\times (n-1)$ circulant matrix over $R$ and $\alpha,\beta,\gamma\in R$ such that $\gamma=\beta$ or $\gamma=-\beta$. Then the matrix $[I_n|B]$ generates a Euclidean isodual code $\mathcal {C}$ of length $2n$ over $R$.}
\begin{proof}
Let $\mathcal {C}$ and $\mathscr{C}$ be the codes generated by $G=[I_n|B]$ and $\mathcal {G}=[-B^T|I_n]$, respectively. Then $\mathscr{C}=\mathcal {C}^\perp$. By the same method was done in construction B, we have the parts of $G$ and $\mathcal {G}$ except $\beta$ and $\gamma$ can be made equivalent. Multiplying all the columns except the $I_n$ by $-1$, we have $\mathscr{C}$ is equivalent to a code $\mathcal {D}$ generated by the matrix $[I_n|A]$, where
$$A=\left[\begin{array}{cccc}\alpha&\gamma&\cdots&\gamma\\ \beta& & &\\\vdots& & \huge{M} &\\\beta& & & \end{array}\right].$$
If $\beta=\gamma$, then $\mathcal {C}=\mathcal {D}$. If $\beta=-\gamma$, then multiply all but the first row of $G$ by $-1$, we have $\mathcal {C}$ is equivalent to $\mathcal {D}$. Therefore $\mathcal {C}$ is equivalent to $\mathscr{C}$.
\end{proof}

Let $\mathcal {C}$ be a Euclidean isodual code over $R$. Then, by Theorem 2, $\Phi(\mathcal {C})$ is $\mathbb{Z}_4$-Euclidean isodual. Therefore, Construction A, B, C and the Gray map $\Phi$ lead to constructions of $\mathbb{Z}_4$-Euclidean isodual actually.
\subsection{Hermitian self-dual codes and Complex lattices}
In this subsection, we will discuss how to construct the unimodular complex lattice from the Hermitian self-dual code over $R$. Firstly, we need another inner product on $R^n$ called Hermitian inner-product.
\begin{definition}
Let $\textbf{w}, \textbf{u}\in R^n$. Then the Hermitian inner-product of $\textbf{w}, \textbf{u}$ is defined as $\langle \textbf{w}, \textbf{u}\rangle=\sum_{i=0}^{n-1}w_i\overline{u}_i$, where $\overline{v}=1-v$. For any code $\mathcal {C}$ of length $n$ over $R$, the Hermitian dual of $\mathcal {C}$ is $\mathcal {C}^H=\{\textbf{w}\in R^n|~\langle \textbf{w}, \textbf{u}\rangle=0, ~\forall \textbf{u}\in\mathcal {C}\}$.
\end{definition}

Let $2\ell+1$ be a square free integer with $\ell\equiv7({\rm mod~}8)$. Define $K=\mathbb{Q}(\sqrt{-2\ell-1})$. Let $\omega=\frac{1+\sqrt{-2\ell-1}}{2}$. Define $\mathcal {O}_K=\mathbb{Z}[\omega]$, we have that $\mathcal {O}_K$ is the ring of integers of the field $K$. The $\omega$ satisfies the equation $X^2-X+\frac{\ell+1}{2}$. Notice that $\ell\equiv7({\rm mod~}8)$ so that $\frac{\ell+1}{2}$ is an integer divisible by $4$.
\par
Consider the canonical homomorphism $\rho:~\mathcal {O}_K\rightarrow \mathcal {O}_K/4\mathcal {O}_K$. Now the image of $\omega$ satisfies the equation $X^2-X=0$.
\begin{lemma}
The ring $\mathcal {O}_K/4\mathcal {O}_K$ is ring isomorphic to $R$.
\end{lemma}
\begin{proof}
Define the map $\Psi:~\mathcal {O}_K/4\mathcal {O}_K\rightarrow R$ by $\Psi(a+b\omega)=a+bv$, where $a, b \in \mathbb{Z}_4$. The map is bijective and the fact that it is a homomorphism follows that $\omega^2=\omega$ in $\mathcal {O}_K/4\mathcal {O}_K$.
\end{proof}

Furthermore, we notice that $\Psi(\overline{\omega})=\Psi(\overline{\frac{1}{2}+\frac{\sqrt{-2\ell-1}}{2}})=\Psi(\overline{\frac{1}{2}+\frac{\sqrt{2\ell+1}}{2}i})=
\Psi(\frac{1}{2}-\frac{\sqrt{2\ell+1}}{2}i)=1-v=\overline{v}$. Therefore complex conjugation corresponds to conjugation in $R$ via the isomorphism $\Psi$.
\par
A lattice $\Lambda$ over $K$ is an $\mathcal {O}_K$-submodule of $K^n$ with full rank. The Hermitian dual of $\Lambda$ is defined as
\begin{equation*}
\Lambda^*=\{\textbf{v}\in K^n|~\langle \textbf{v}, \textbf{w}\rangle\in \mathcal {O}_k,~\forall \textbf{w}\in\Lambda\}.
\end{equation*}

If $\Lambda=\Lambda^*$, we say $\Lambda$ is unimodular and if $\Lambda\subseteq \Lambda^*$, we say $\Lambda$ is integral.
\begin{lemma}
Let $\mathcal {C}$ be a linear code of length $n$ over $R$. Then we have the following results\\
{\rm (i)}~$\Lambda(\mathcal {C})=\{\textbf{v}\in \mathcal {O}_K^n|~\rho(\textbf{v})\in \mathcal {C}\}$ is an $\mathcal {O}_K$-lattice. \\
{\rm (ii)}~$\Lambda(\mathcal {C}^H)=4\Lambda(\mathcal {C})^*$.\\
{\rm (iii)}~$(\frac{1}{2}\Lambda(\mathcal {C}))^*=2\Lambda(\mathcal {C})^*$.
\end{lemma}
\begin{proof}
(i)~It is immediate from the definition of $\Lambda(\mathcal {C})$ and $\mathcal {C}$ is an $R$-submodule of $R^n$.
\par\vskip 2mm
(ii)~If $\textbf{v}\in4\Lambda(\mathcal {C})^*$, then $\langle \frac{1}{4}\textbf{v}, \textbf{w}\rangle\in \mathcal {O}_K$ for all $\textbf{w}\in \Lambda(\mathcal {C})$. Therefore, we have $\sum_{i=0}^{n-1}\frac{1}{4}v_i\overline{w}_i\in \mathcal {O}_K\Rightarrow \sum_{i=0}^{n-1}v_i\overline{w}_i\in 4\mathcal {O}_k\Rightarrow \langle \rho(\textbf{v}), \rho(\textbf{w})\rangle=0$, which implies that $\textbf{v}\in \Lambda(\mathcal {C}^H)$. Then $4\Lambda(\mathcal {C})\subseteq \Lambda(\mathcal {C}^H)$.
\par
Let $\textbf{v}\in \Lambda(\mathcal {C}^H)$. Then $\rho(\textbf{v})\in \mathcal {C}^H$ and $\langle \rho(\textbf{v}), \rho(\textbf{w})\rangle=0$ for all $\textbf{w}\in \Lambda(\mathcal {C})$. Then we have $\sum_{i=0}^{n-1}v_i\overline{w}_i\in 4\mathcal {O}_K\Rightarrow \sum_{i=0}^{n-1}\frac{1}{4}v_i\overline{w}_i\in \mathcal {O}_K\Rightarrow \langle \frac{1}{4}\textbf{v}, \textbf{w}\rangle\in \mathcal {O}_K$, which implies that $\textbf{v}\in 4 \Lambda(\mathcal {C})^*$. Therefore $\Lambda(\mathcal {C}^H)=4\Lambda(\mathcal {C})^*$.
\par\vskip 2mm
(iii)~Let $\textbf{v}\in (\frac{1}{2}\Lambda(\mathcal {C}))^*$, that is $\langle \textbf{v}, \textbf{w}\rangle\in \mathcal {O}_K$ for all $\textbf{w}\in \frac{1}{2}\Lambda(\mathcal {C})$. This implies that $(\frac{1}{2}\times 2)\langle \textbf{v}, \textbf{w}\rangle \in \mathcal {O}_K$. Then we have $\langle \frac{1}{2}\textbf{v}, 2\textbf{w}\rangle\in \mathcal {O}_K$ for all $\textbf{w}\in \frac{1}{2}\Lambda(\mathcal {C})$, that is for all $2\textbf{w}\in \Lambda(\mathcal {C})$. Then we have $\frac{1}{2}\textbf{v}\in \Lambda(\mathcal {C})^*$, which implies that $\textbf{v}\in 2\Lambda(\mathcal {C})^*$. Therefore $(\frac{1}{2}\Lambda(\mathcal {C}))^*\subseteq 2 \Lambda(\mathcal {C})^*$.
\par
Now, assume that $\textbf{v}\in 2\Lambda(\mathcal {C})^*$. Then $\langle \textbf{v}, \textbf{w}\rangle\in \mathcal {O}_K$ for all $\textbf{w}\in 2\Lambda(\mathcal {C})$. That is $\frac{1}{2}\langle \textbf{v}, \textbf{w}\rangle\in \mathcal {O}_K$ for all $\textbf{w}\in \Lambda(\mathcal {C})$, which implies that $\textbf{v}\in (\frac{1}{2}\Lambda(\mathcal {C}))^*$. Then $2\Lambda(\mathcal {C})^*=(\frac{1}{2}\Lambda(\mathcal {C}))^*$.
\end{proof}
\begin{theorem}
The linear code $\mathcal {C}$ over $R$ is Hermitian self-dual if and only if $\frac{1}{2}\Lambda(\mathcal {C})$ is unimodular.
\end{theorem}
\begin{proof}
If $\mathcal {C}=\mathcal {C}^H$, then by Lemma 4(iii), $(\frac{1}{2}\Lambda(\mathcal {C}))^*=2\Lambda(\mathcal {C})^*$. Furthermore, by Lemma 4(ii), we have $\Lambda(\mathcal {C}^H)=4\Lambda(\mathcal {C})^*$, which implies that $2\Lambda(\mathcal {C})^*=\frac{1}{2}\Lambda(\mathcal {C}^H)=\frac{1}{2}\Lambda(\mathcal {C})$. Therefore $\frac{1}{2}\Lambda(\mathcal {C})=(\frac{1}{2}\Lambda(\mathcal {C}))^*$.
\par
Next, let $\frac{1}{2}\Lambda(\mathcal {C})=(\frac{1}{2}\Lambda(\mathcal {C}))^*$. Then $(\frac{1}{2}\Lambda(\mathcal {C}))^*=2\Lambda(\mathcal {C})^*$ by Lemma 4(iii). Furthermore, $2\Lambda(\mathcal {C})^*=\frac{1}{2}\Lambda(\mathcal {C}^H)$ by Lemma 4(ii). Therefore, we have $\frac{1}{2}\Lambda(\mathcal {C}^H)=\frac{1}{2}\Lambda(\mathcal {C})$. In the following, we show $\mathcal {C}=\mathcal {C}^H$. Let $\textbf{v}\in \mathcal {C}$. Then there exists $\textbf{w}\in \Lambda(\mathcal {C})$ such that $\rho(\textbf{w})=\textbf{v}$. But $\Lambda(\mathcal {C})=\Lambda(\mathcal {C}^H)$, which implies that $\rho(\textbf{w})\in \mathcal {C}^H$. This yields $\mathcal {C}\subseteq \mathcal {C}^H$. Similarly, we can prove $\mathcal {C}^H\subseteq \mathcal {C}$. Thus $\mathcal {C}=\mathcal {C}^H$ implying $\mathcal {C}$ is Hermitian self-dual.
\end{proof}
\subsection{MDS codes over $R$}
In this subsection, we discuss another important class of linear codes over $R$ called MDS codes. For any Frobenius ring $R$, the Singleton bound for a code of length $n$ over $R$ states that
\begin{equation*}
d_{H}(\mathcal {C})\leq n-\log_{|R|}{|\mathcal {C}|}+1,
\end{equation*}
where $d_H(\mathcal {C})$ denotes the minimum Hamming distance of $\mathcal {C}$. A code meeting this bound is said to be a MDS code over $R$.
\begin{theorem}
Let $\mathcal {C}=v\mathcal {C}_1\oplus (1-v)\mathcal {C}_2$ be a linear code of length $n$ over $R$. Then we have\\
{\rm (i)}~$d_H(\mathcal {C})={\rm min}\{ d_H(\mathcal {C}_1), d_H(\mathcal {C}_2)\}$;\\
{\rm (ii)}~$\mathcal {C}$ is an $(n, M, d)$ MDS code over $R$ if and only if $\mathcal {C}_1$ and $\mathcal {C}_2$ are both $(n, \sqrt{M}, d)$ MDS codes over $\mathbb{Z}_4$.
\end{theorem}
\begin{proof}
(i)~It is straightforward from the fact that for any codeword $\textbf{c}=v\textbf{c}_1+(1-v)\textbf{c}_2\in \mathcal {C}$, $\textbf{c}=\textbf{0}$ if and only if $\textbf{c}_1=\textbf{c}_2=\textbf{0}$.
\par\vskip 2mm
(ii)~Denote $d_H^{(1)}(\mathcal {C})$ and $d_H^{(2)}(\mathcal {C})$ as the minimum Hamming distances of $\mathcal {C}_1$ and $\mathcal {C}_2$, respectively. If $d_H(\mathcal {C})=d_H^{(1)}(\mathcal {C})$, then $d_H^{(2)}(\mathcal {C})\geq d_H^{(1)}(\mathcal {C})$ by (i). Let $\mathcal {C}$ be an $(n, M, d)$ MDS code. Then $d=n-\log_{16}|\mathcal {C}|$+1. Let $M_1$ and $M_2$ be the codewords number of $\mathcal {C}_1$ and $\mathcal {C}_2$, respectively. Then, by the Singleton bound, we have
\begin{equation*}
d_H^{(1)}\leq n-\log_{4}M_1+1
\end{equation*}
and
\begin{equation*}
d_H^{(2)}\leq n-\log_{4}M_2+1.
\end{equation*}
From $d=d_H^{(1)}\leq d_H^{(2)}$, we have that
\begin{equation}
\log_4\sqrt{M}\geq \log_4M_1
\end{equation}
and
\begin{equation}
\log_4\sqrt{M}\geq \log_4M_2.
\end{equation}
Therefore the equalities in the above equations (1) and (2) hold if and only if $M_1=M_2=\sqrt{M}$. From the Singleton bound and $\mathcal {C}$ is an MDS code, we deduce $\mathcal {C}_1$ and $\mathcal {C}_2$ are both MDS codes with the same parameters. The necessary part is straightforward by the Singleton bound.
\end{proof}
\begin{corollary}
There are no non-trivial MDS codes over $R$.
\end{corollary}
\begin{proof}
By Theorem 10(ii), we know that there exist non-trivial MDS codes over $R$ if and only if there exist non-trivial MDS codes over $\mathbb{Z}_4$. But it is well known that there are no non-trivial MDS codes over $\mathbb{Z}_4$.
\end{proof}

In the following, we give the bound of Gray distance of the linear code over $R$.
\begin{theorem}
Let $\mathcal {C}=v\mathcal {C}_1\oplus (1-v)\mathcal {C}_2$ be a linear code of length $n$ over $R$. Then we have\\
{\rm (i)}~$d_G(\mathcal {C})={\rm min}\{ d_L(\mathcal {C}_1), d_L(\mathcal {C}_2)\}$;\\
{\rm (ii)}~$d_G(\mathcal {C})\leq 2n-\log_2^{{\rm min}\{|\mathcal {C}_1|, |\mathcal {C}_2|\}}+1$.
\end{theorem}
\begin{proof}
(i)~Clearly, $d_G(\mathcal {C})={\rm min}\{ d_G(v\mathcal {C}_1), d_G((1-v)\mathcal {C}_2)\}={\rm min}\{ d_L(\Phi(v\mathcal {C}_1)), d_L(\Phi((1-v)\mathcal {C}_2))\}$. Denote by $*$ the componentwise multiplication of two vectors, i.e. $$(x_1, x_2, \ldots, x_n)*(y_1, y_2, \ldots, y_n)=(x_1y_1, x_2y_2, \ldots, x_ny_n).$$ Then, by the definition of the Gray map $\Phi$, we have $\Phi(v\mathcal {C}_1)=(0,1)*\Phi(\mathcal {C}_1)$ and $\Phi((1-v)\mathcal {C}_2)=(1,0)*\Phi(\mathcal {C}_2)$, which implies that $d_L(\mathcal {C}_1)=d_L(\Phi(v\mathcal {C}_1))$ and $d_L((1-v)\mathcal {C}_2)=d_L(\Phi((1-v)\mathcal {C}_2))$ respectively.\\
(ii)~From the Hamming weight Singleton bound for binary codes, we have $$d_L(\mathcal {C}_i)\leq 2n-\log_2^{|\mathcal {C}_i|}+1$$ for $i=1,2$. Then, by (i), we have $$d_G(\mathcal {C})\leq 2n-\log_2^{{\rm min}\{|\mathcal {C}_1|, |\mathcal {C}_2|\}}+1.$$
\end{proof}

We shall refer to codes meeting the bound in Theorem 11(ii) as \emph{maximal Gray distance separable (MGDS) codes}. Clearly, $\mathcal {C}$ is a MGDS code over $R$ if and only if $\mathcal {C}_1$ and $\mathcal {C}_2$ are both quaternary \emph{maximal Lee distance separable (MLDS) codes} and with the same parameters. A quaternary code $\mathscr{C}$ is called MLDS code if $d_L(\mathscr{C})= 2n-\log_2^{|\mathscr{C}|}+1$. Therefore if $\mathcal {C}$ is a MGDS code of length $n$ over $R$, then $\mathcal {C}$ is either $(\textbf{2})$, or the whole space, where the symbol $\textbf{2}$ denote the all $2$-vectors of length $n$ over $R$.
\section{Cyclic codes over $R$}
As a special class of linear codes, cyclic codes play very important roles in the coding theory. In this section, we give some useful results on cyclic codes over $R$.
\par
Let $T$ be the cyclic shift operator on $R^n$, i.e. for any vector $\textbf{c}=(c_0, c_1, \ldots, c_{n-1})$ of $ R^n$, $T(\textbf{c})=(c_{n-1}, c_0, \ldots, c_{n-2})$.
\par
A linear code $\mathcal {C}$ of length $n$ over $R$ is called cyclic if and only if $T(\mathcal {C})=\mathcal {C}$.
Define the polynomial ring $R_n=R[X]/(X^n-1)=\{c_0+c_1X+\cdots+c_{n-1}X^{n-1}+(X^n-1)|~c_0, c_1, \ldots, c_{n-1}\in R\}$. For any polynomial $c(X)+(X^n-1)\in R_n$, we denote it as $c(X)$ for simplicity.
\par
Define a map as follows
\begin{equation*}
\begin{split}
\varphi:~~R^n&\rightarrow R_n=R[X]/(X^n-1)\\
 (c_0, c_1, \ldots, c_{n-1})&\mapsto c(X)=c_0+c_1X+\cdots +c_{n-1}X^{n-1}.
 \end{split}
 \end{equation*}
Clearly, $\varphi$ is an $R$-module isomorphism from $R^n$ to $R_n$. And a linear code $\mathcal {C}$ of length $n$ is cyclic over $R$ if and only if $\varphi(\mathcal {C})$ is an ideal of $R_n$. Sometimes, we identify the cyclic code $\mathcal {C}$ to the ideal of $R_n$.
\begin{theorem}
A linear code $\mathcal {C}=v\mathcal {C}_1\oplus (1-v)\mathcal {C}_2$ is cyclic over $R$ if and only if $\mathcal {C}_1$ and $\mathcal {C}_2$ are both cyclic over $\mathbb{Z}_4$.
\end{theorem}
\begin{proof}
Let $(a_0, a_1, \ldots, a_{n-1}) \in \mathcal {C}_1$ and $(b_0, b_1, \ldots, b_{n-1})\in \mathcal {C}_2$. Assume that $c_i=va_i+(1-v)b_i$ for $i=0,1,\ldots, n-1$. Then the vector $(c_0, c_1, \ldots, c_{n-1})$ belongs to $\mathcal {C}$. Since $\mathcal {C}$ is a cyclic code, it follows that $(c_{n-1}, c_0, \ldots, c_{n-2})\in \mathcal {C}$. Note that $(c_{n-1}, c_0, \ldots, c_{n-2})=v(a_{n-1}, a_0, \ldots, a_{n-2})+(1-v)(b_{n-1}, b_0, \ldots, b_{n-2})$. Hence $(a_{n-1}, a_0, \ldots, a_{n-2})\in \mathcal {C}_1$ and $(b_{n-1}, b_0, \ldots, b_{n-2})\in \mathcal {C}_2$, which implies that $\mathcal {C}_1$ and $\mathcal {C}_2$ are both cyclic codes over $\mathbb{Z}_4$.
\par
Conversely, let $\mathcal {C}_1$ and $\mathcal {C}_2$ be both cyclic codes over $\mathbb{Z}_4$. Let $(c_0, c_1, \ldots, c_{n-1})\in \mathcal {C}$, where $c_i=va_i+(1-v)b_i$ for $i=0,1,\ldots,n-1$. Then $(a_0, a_1, \ldots, a_{n-1})\in\mathcal {C}_1$ and $(b_0, b_1, \ldots, b_{n-1})\in\mathcal {C}_2$. Note that $(c_{n-1}, c_0, \ldots, c_{n-2})=v(a_{n-1}, a_0, \ldots, a_{n-2})+(1-v)(b_{n-1}, b_0, \ldots, b_{n-2})\in v\mathcal {C}_1\oplus(1-v)\mathcal {C}_2=\mathcal {C}$. Therefore, $\mathcal {C}$ is a cyclic code over $R$.
\end{proof}

In the following of this section, we assume that $n$ is an odd positive integer. Let $C$ be a cyclic code of length $n$ over $\mathbb{Z}_4$. Then there exist unique monic polynomials $f(X), g(X), h(X)$ such that $X^n-1=f(X)g(X)h(X)$ and $C=(f(X)g(X))\oplus (2f(X)h(X))$. See \cite{Wan} for the details.
\begin{theorem}
Let $\mathcal {C}=v\mathcal {C}_1\oplus (1-v)\mathcal {C}_2$ be a cyclic code of length $n$ over $R$. Then $\mathcal {C}=(vf_1(X)g_1(X)+(1-v)f_2(X)g_2(X))\oplus (2vf_1(X)h_1(X)+2(1-v)f_2(X)h_2(X))$, where $f_1(X)g_1(X)h_1(X)=f_2(X)g_2(X)h_2(X)=X^n-1$ and $\mathcal {C}_1=(f_1(X)g_1(X))\oplus (2f_1(X)h_1(X))$, $\mathcal {C}_2=(f_2(X)g_2(X))\oplus (2f_2(X)h_2(X))$ over $\mathbb{Z}_4$, respectively.
\end{theorem}
\begin{proof}
Let $\widetilde{\mathcal {C}}=(vf_1(X)g_1(X)+(1-v)f_2(X)g_2(X))\oplus (2vf_1(X)h_1(X)+2(1-v)f_2(X)h_2(X))$, $\mathcal {C}_1=(f_1(X)g_1(X))\oplus (2f_1(X)h_1(X))$ and $\mathcal {C}_2=(f_2(X)g_2(X))\oplus (2f_2(X)h_2(X))$. Clearly, $\widetilde{\mathcal {C}}\subseteq \mathcal {C}$. For $v\mathcal {C}_1$, we have $v\mathcal {C}_1=v\mathcal {C}$ since $v^2=v$ over $\mathbb{Z}_4$. Similarly, $(1-v)\mathcal {C}_2=(1-v)\mathcal {C}$. Therefore $v\mathcal {C}_1\oplus (v-1)\mathcal {C}_2\subseteq \mathcal {C}$. Thus $\mathcal {C}=\widetilde{\mathcal {C}}$.
\end{proof}
\begin{corollary}
The quotient polynomial ring $R[X]/(X^n-1)$ is principal.
\end{corollary}
\begin{proof}
Let $C=(f(X)g(X))\oplus (2f(X)h(X))$ be a cyclic code of length $n$ over $\mathbb{Z}_4$, where $X^n-1=f(X)g(X)h(X)$. Then $C=(f(X)g(X)+2f(X))$. (See Theorem 7.25 and Theorem 7.26 in \cite{Wan} for the details.) By Theorem 13, we have  any cyclic code $\mathcal {C}$ is principal over $R$, which implies the results.
\end{proof}

Furthermore, the number of distinct cyclic codes of odd length $n$ over $R$ is $9^r$, where $r$ is the number of the basic irreducible factors of $X^n-1$ over $\mathbb{Z}_4$.
\par We have observed numerous times that Euclidean cyclic self-dual codes over $R$ exist. (See Example 2 in Section 6.) Theorem 13 gives the generating polynomials for cyclic codes over $R$. The next result gives the conditions on these polynomials that lead to Euclidean cyclic self-dual codes.
\begin{theorem}
Let $\mathcal {C}=(vf_1(X)g_1(X)+(1-v)f_2(X)g_2(X))\oplus (2vf_1(X)h_1(X)+2(1-v)f_2(X)h_2(X))$, where $f_1(X)g_1(X)h_1(X)=f_2(X)g_2(X)h_2(X)=X^n-1$ and $\mathcal {C}_1=(f_1(X)g_1(X))\oplus (2f_1(X)h_1(X))$, $\mathcal {C}_2=(f_2(X)g_2(X))\oplus (2f_2(X)h_2(X))$ over $\mathbb{Z}_4$, respectively. Then $\mathcal {C}$ is Euclidean self-dual if and only if $f_1(X)=h_1^*(X), g_1(X)=g_1^*(X)$ and $f_2(X)=h_2^*(X), g_2(X)=g_2^*(X)$, where $f^*(X)=X^{{\rm deg}f(X)}f(X^{-1})$.
\end{theorem}
\begin{proof}
Firstly, by $\mathcal {C}^\perp=v\mathcal {C}_1^\perp\oplus (v-1)\mathcal {C}_2^\perp$, we have $\mathcal {C}^\perp$ is also a cyclic code if $\mathcal {C}$ is a cyclic code. Moreover, by Theorem 4, we have $\mathcal {C}$ is Euclidean self-dual over $R$ if and only if $\mathcal {C}_1$ and $\mathcal {C}_2$ are both Euclidean self-dual over $\mathbb{Z}_4$. Then, by Theorem 12.5.10 in \cite{Huffuman}, we deduce the result.
\end{proof}

When do there exist non-zero Euclidean cyclic self-dual codes of odd length $n$ over $R$? By Theorem 4 and Theorem 3 \cite{Pless1}, we give an answer about this problem.
\begin{theorem}
Non-zero Euclidean cyclic self-dual codes of odd length $n$ exist over $R$ if and only if $2^j\not\equiv -1~({\rm mod}n)$ for any $j$.
\end{theorem}

For example, if $n=7$, then $n$ satisfies the condition in Theorem 14. And then, there exist non-zero Euclidean self-dual codes of length $7$ over $R$. The Example 2 in  Section 6 shows that there exist non-zero Euclidean cyclic self-dual codes of length $7$ over $R$ indeed.
\par
In the following, we consider some properties of the generating idempotents of cyclic codes over $R$. An element $e(X)\in \mathcal {C}$ is called an idempotent element if $e(X)^2=e(X)$ in $R_n$.
\begin{theorem}
Let $\mathcal {C}$ be a cyclic code of odd length $n$. Then there exists a unique idempotent element $e(X)=ve_1(X)+(1-v)e_2(X)\in R[X]$ such that $\mathcal {C}=(e(X))$.
\end{theorem}
\begin{proof}
If $n$ is odd, then there exist unique idempotent elements $e_1(X), e_2(X) \in \mathbb{Z}_4[X]$ such that $\mathcal {C}_1=(e_1(X))$ and $\mathcal {C}_2=(e_2(X))$. By Theorem 13, we have $\mathcal {C}=(ve_1(X)+(1-v)e_2(X))$. Let $e(X)=ve_1(X)+(1-v)e_2(X)$. Then $e(X)^2=ve_1(X)^2+(1-v)e_2(X)^2=ve_1(X)+(1-v)e_2(X)=e(X)$, which implies that $e(X)$ is an idempotent element of $\mathcal {C}$. If there is another $d(X)\in \mathcal {C}$ such that $\mathcal {C}=(d(X))$ and $d(X)^2=d(X)$. Since $d(X)\in \mathcal {C}=(e(X))$, we have that $d(X)=a(X)e(X)$ for some $a(X)\in R_n$. And then, $d(X)e(X)=a(X)e(X)^2=d(X)$. Similarly, we can prove $d(X)e(X)=e(X)$, which implies that $e(X)$ is unique.
\end{proof}

The idempotent element $e(X)$ in above theorem is called the generating idempotent of $\mathcal {C}$.
\begin{theorem}
Let $\mathcal {C}=v\mathcal {C}_1\oplus (1-v)\mathcal {C}_2$ be a cyclic code of length $n$ over $R$. Let $e(X)=ve_1(X)+(1-v)e_2(X)$, where $e_1(X)$ and $e_2(X)$ are generating idempotents of $\mathcal {C}_1$ and $\mathcal {C}_2$ over $\mathbb{Z}_4$, respectively. Then the Euclidean dual code $\mathcal {C}^\perp$ has $1-e(X^{-1})$ as its generating idempotent.
\end{theorem}
\begin{proof}
By Theorem 4, we have $\mathcal {C}^\perp=v\mathcal {C}_1^\perp \oplus (v-1)\mathcal {C}_2^\perp$. Moreover, $\mathcal {C}^\perp$ is also a cyclic code since $\mathcal {C}_1^\perp$ and $\mathcal {C}_2^\perp$ are both cyclic codes. Let $e_1(X)$ and $e_2(X)$ be generating idempotents of $\mathcal {C}_1$ and $\mathcal {C}_2$, respectively. Then $\mathcal {C}_1^\perp$ and $\mathcal {C}_2^\perp$ have $1-e_1(X^{-1})$ and $1-e_2(X^{-1})$ as their generating idempotents respectively. (See Lemma 12.3.23(i) in \cite{Huffuman} for the details.) Let $\widetilde{e}(X)$ be the generating idempotent of $\mathcal {C}^\perp$. Then, by Theorem 16, $\widetilde{e}(X)=v(1-e_1(X^{-1}))+(1-v)(1-e_2(X^{-1}))=1-e(X^{-1})$.
\end{proof}

\section{Quadratic residue codes over $R$}
In this section, let $p$ be a prime number with $p\equiv \pm 1 ({\rm mod}8)$. Let $\mathcal {Q}_p$ denote the set of nonzero quadratic residues modulo $p$, and let $\mathcal {N}_p$ be the set of quadratic non-residues modulo $p$.
\par
Let $Q(X)=\sum_{i\in \mathcal {Q}_p}X^i$, $N(X)=\sum_{i\in \mathcal {N}_p}X^i$ and $J(X)=p\sum_{i=0}^{p-1}X^i$. By Theorem 16 and Theorem 8 \cite{Pless2}, we have the following results immediately.
\begin{lemma}
Define $r$ by $p=8r\pm 1$. If $r$ is odd, denote the set $\mathcal {S}_0=\{ Q(X)+2N(X), N(X)+2Q(X), 1-Q(X)+2N(X), 1-N(X)+2Q(X)\}$. If $r$ is even, denote the set $\mathcal {S}_e=\{-Q(X), -N(X), 1+Q(X), 1+N(X)\}$. Then\\
{\rm (i)}~For any $e_1(X), e_2(X)\in \mathcal {S}_0$ or $e_1(X), e_2(X)\in \mathcal {S}_e$, we have $e(X)=ve_1(X)+(1-v)e_2(X)$ is the idempotent of $R_p$.\\
{\rm (ii)}~$J(X)$ is an idempotent of $R_p$.
\end{lemma}

We now discuss the quadratic residue codes over $R$. Firstly, we give the definitions of these codes. The definitions depend upon the value $p$ modulo $8$.
 \vskip 3mm \noindent
   {\bf Case I:~~$p\equiv -1({\rm mod}8)$}
\begin{definition}
Let $p+1=8r$. If $r$ is odd, define
\begin{equation*}
\mathcal {D}_1=(v(Q(X)+2N(X))+(1-v)(N(X)+2Q(X))),
\end{equation*}
\begin{equation*}
\mathcal {D}_2=(v(N(X)+2Q(X))+(1-v)(Q(X)+2N(X))),
\end{equation*}
and
\begin{equation*}
\mathcal {E}_1=(v(1-N(X)+2Q(X))+(1-v)(1-Q(X)+2N(X))),
\end{equation*}
\begin{equation*}
\mathcal {E}_2=(v(1-Q(X)+2N(X))+(1-v)(1-N(X)+2Q(X))).
\end{equation*}
If $r$ is even, define
\begin{equation*}
\mathcal {D}_1=(v(-Q(X))+(1-v)(-N(X))),
\end{equation*}
\begin{equation*}
\mathcal {D}_2=(v(-N(X))+(1-v)(-Q(X))),
\end{equation*}
and
\begin{equation*}
\mathcal {E}_1=(v(1+N(X))+(1-v)(1+Q(X))),
\end{equation*}
\begin{equation*}
\mathcal {E}_2=(v(1+Q(X))+(1-v)(1+N(X))).
\end{equation*}
These cyclic codes of length $p$ are called the quadratic residue codes over $R$ at the case I.
\end{definition}
\par\vskip 2mm
Let $a$ be a non-zero positive integer defined as $\mu_a(i)=ai$ for any positive integer $i$. This map acts on polynomials as
\begin{equation*}
\mu_a(\sum_iX^i)=\sum_iX^{ai}.
\end{equation*}
\begin{theorem}
Let $p\equiv-1({\rm mod}8)$. Then the quadratic residue codes defined above satisfy the following:\\
{\rm (i)}~$\mathcal {D}_i\mu_a=\mathcal {D}_i$ and $\mathcal {E}_i\mu_a=\mathcal {E}_i$ for $i=1,2$ and $a\in \mathcal {Q}_p$; $\mathcal {D}_1\mu_a=\mathcal {D}_2$ and $\mathcal {E}_1\mu_a=\mathcal {E}_2$ for $a\in \mathcal {N}_p$.\\
{\rm (ii)}~$\mathcal {D}_1\cap \mathcal {D}_2=(J(X))$ and $\mathcal {D}_1+\mathcal {D}_2=R_p$.\\
{\rm (iii)}~$\mathcal {E}_1\cap\mathcal {E}_2=\{0\}$ and $\mathcal {E}_1 + \mathcal {E}_2=(J(X))^\perp$.\\
{\rm (iv)}~$|\mathcal {D}_1|=|\mathcal {D}_2|=4^{p+1}$ and $|\mathcal {E}_1|=|\mathcal {E}_2|=4^{p-1}$.\\
{\rm (v)}~$\mathcal {D}_i=\mathcal {E}_i+(J(X))$ for $i=1,2$.\\
{\rm (vi)}~$\mathcal {E}_1$ and $\mathcal {E}_2$ are Euclidean self-orthogonal and $\mathcal {E}_i^\perp=\mathcal {D}_i$ for $i=1,2$.
\end{theorem}
\begin{proof}
Let $p+1=8r$. We only verify when $r$ is odd. The case of $r$ is even can be proved similarly.

(i)~If $a\in \mathcal {Q}_p$, then $(v(Q(X)+2N(X))+(1-v)(N(X)+2Q(X)))\mu_a=v(Q(X)+2N(X))+(1-v)(N(X)+2Q(X))$, which implies that $\mathcal {D}_1\mu_a=\mathcal {D}_1$. Similarly, $\mathcal {D}_2\mu_a=\mathcal {D}_2$.

If $a\in \mathcal {N}_p$, then $(v(Q(X)+2N(X))+(1-v)(N(X)+2Q(X)))\mu_a=v(N(X)+2Q(X))+(1-v)(Q(X)+2N(X))$, which implies that $\mathcal {D}_1\mu_a=\mathcal {D}_2$.

The parts of (i) involving $\mathcal {E}_i$ are similar.

(ii)~Since $p\equiv -1 ({\rm mod}8)$, it follows that $J(X)=3\sum_{i=0}^{p-1}X^i=3+3Q(X)+3N(X)$. Therefore $(v(Q(X)+2N(X))+(1-v)(N(X)+2Q(X)))(v(N(X)+2Q(X))+(1-v)(Q(X)+2N(X)))=(Q(X)+2N(X))(N(X)+2Q(X))=J(X)$, which implies that $\mathcal {D}_1\cap \mathcal {D}_2=(J(X))$. Moreover, $v(Q(X)+2N(X))+(1-v)(N(X)+2Q(X))+v(N(X)+2Q(X))+(1-v)(Q(X)+2N(X))-J(X)=3Q(X)+3N(X)-J(X)=1$, which implies that $\mathcal {D}_1+\mathcal {D}_2=R_p$.

(iii)~For $\mathcal {E}_1\cap \mathcal {E}_2$, we have $(v(1-N(X)+2Q(X))+(1-v)(1-Q(X)+2N(X)))(v(1-Q(X)+2N(X))+(1-v)(1-N(X)+2Q(X)))=(1-N(X)+2Q(X))(1-Q(X)+2N(X))=1+N(X)+Q(X)+J(X)=0$, which implies that $\mathcal {E}_1\cap \mathcal {E}_2=\{0\}$.

For $\mathcal {E}_1+\mathcal {E}_2$, it has generating idempotent $1-N(X)+2Q(X)+1-Q(X)+2N(X)=2+N(X)+Q(X)=1-J(X)=1-J(X)\mu_{-1}$ as $J(X){\mu_{-1}}=J(X)$. Then, by Theorem 17, $\mathcal {E}_1+\mathcal {E}_2=(J(X))^\perp$.

(iv)~We use the fact that $|\mathcal {D}_1+\mathcal {D}_2|=|\mathcal {D}_1||\mathcal {D}_2|/|\mathcal {D}_1\cap \mathcal {D}_2|$. By (i), $|\mathcal {D}_1|=|\mathcal {D}_2|$, and by (ii), $|\mathcal {D}_1+\mathcal {D}_2|=16^p$ and $|\mathcal {D}_1\cap \mathcal {D}_2|=16$. Therefore, $|\mathcal {D}_1|=|\mathcal {D}_2|=16^{(p+1)/2}=4^{p+1}$. Similarly, by (i) and (iii), we can prove $|\mathcal {E}_1|=|\mathcal {E}_2|=4^{p-1}$.

(v)~From (ii), we have $J(X)\in \mathcal {D}_2$ implying that $(v(N(X)+2Q(X))+(1-v)(Q(X)+2N(X)))J(X)=J(X)$ as $v(N(X)+2Q(X))+(1-v)(Q(X)+2N(X))$ is the multiplicative identity of $\mathcal {D}_2$. Then the generating idempotent for $\mathcal {E}_1+(J(X))$ is $v(1-N(X)+2Q(X))+(1-v)(1-Q(X)+2N(X))+J(X)-(v(1-N(X)+2Q(X))+(1-v)(1-Q(X)+2N(X)))J(X)=v(1-N(X)+2Q(X))+(1-v)(1-Q(X)+2N(X))+J(X)+(J(X)-J(X))
=v(Q(X)+2N(X))+(1-v)(N(X)+2Q(X))$, which implies that $\mathcal {E}_1+(J(X))=\mathcal {D}_1$. Similarly, $\mathcal {E}_2+(J(X))=\mathcal {D}_2$.

(vi)~From Theorem 17, the generating idempotent for $\mathcal {E}_1^\perp$ is $1-(v(1-N(X)+2Q(X))+(1-v)(1-Q(X)+2N(X)))\mu_{-1}=v(N(X)+2Q(X))\mu_{-1}+(1-v)(Q(X)+2N(X))\mu_{-1}$. Since $-1\in \mathcal {N}_p$ as $p\equiv -1({\rm mod}8)$, it follows that $N(X)\mu_{-1}=Q(X)$ and $Q(X)\mu_{-1}=N(X)$. Therefore the generating idempotent for $\mathcal {E}_1^\perp$ is $v(Q(X)+2N(X))+(1-v)(N(X)+2Q(X))$ implying that $\mathcal {E}_1^\perp=\mathcal {D}_1$. Similarly, $\mathcal {E}_2^\perp=\mathcal {D}_2$. From (v), we have $\mathcal {E}_i\subseteq \mathcal {D}_i$ implying that $\mathcal {E}_i$ is Euclidean self-orthogonal for $i=1,2$.
\end{proof}
\vskip 3mm \noindent
   {\bf Case II:~~$p\equiv1({\rm mod}8)$}
\begin{definition}
Let $p-1=8r$. If $r$ is odd, define
\begin{equation*}
\mathcal {D}_1=(v(1-N(X)+2Q(X))+(1-v)(1-Q(X)+2N(X)))
\end{equation*}
\begin{equation*}
\mathcal {D}_2=(v(1-Q(X)+2N(X))+(1-v)(1-N(X)+2Q(X)))
\end{equation*}
and
\begin{equation*}
\mathcal {E}_1=(v(Q(X)+2N(X))+(1-v)(N(X)+2Q(X)))
\end{equation*}
\begin{equation*}
\mathcal {E}_1=(v(N(X)+2Q(X))+(1-v)(Q(X)+2N(X))).
\end{equation*}
If $r$ is even, define
\begin{equation*}
\mathcal {D}_1=(v(1+N(X))+(1-v)(1+Q(X))),
\end{equation*}
\begin{equation*}
\mathcal {D}_2=(v(1+Q(X))+(1-v)(1+N(X))),
\end{equation*}
and
\begin{equation*}
\mathcal {E}_1=(v(-Q(X))+(1-v)(-N(X))),
\end{equation*}
\begin{equation*}
\mathcal {E}_2=(v(-N(X))+(1-v)(-Q(X))).
\end{equation*}
These cyclic codes of length $p$ are called the quadratic residue codes over $R$ at the case II.
\end{definition}

Similar to Theorem 18, we also have the following results. Here we omit the proof.
\begin{theorem}
Let $p\equiv1({\rm mod}8)$. Then the quadratic residue codes defined above satisfy the following:\\
{\rm (i)}~$\mathcal {D}_i\mu_a=\mathcal {D}_i$ and $\mathcal {E}_i\mu_a=\mathcal {E}_i$ for $i=1,2$ and $a\in \mathcal {Q}_p$; $\mathcal {D}_1\mu_a=\mathcal {D}_2$ and $\mathcal {E}_1\mu_a=\mathcal {E}_2$ for $a\in \mathcal {N}_p$.\\
{\rm (ii)}~$\mathcal {D}_1\cap \mathcal {D}_2=(J(X))$ and $\mathcal {D}_1+\mathcal {D}_2=R_p$.\\
{\rm (iii)}~$\mathcal {E}_1\cap\mathcal {E}_2=\{0\}$ and $\mathcal {E}_1+ \mathcal {E}_2=(J(X))^\perp$.\\
{\rm (iv)}~$|\mathcal {D}_1|=|\mathcal {D}_2|=4^{p+1}$ and $|\mathcal {E}_1|=|\mathcal {E}_2|=4^{p-1}$.\\
{\rm (v)}~$\mathcal {D}_i=\mathcal {E}_i+(J(X))$ for $i=1,2$.\\
{\rm (vi)}~$\mathcal {E}_1^\perp=\mathcal {D}_2$ and  $\mathcal {E}_2^\perp=\mathcal {D}_1$.
\end{theorem}

Let $\mathcal {D}_1$ and $\mathcal {D}_2$ be the quadratic residue codes defined above. In the following, we discuss two extensions of $\mathcal {D}_i$ denoted as $\widehat{\mathcal {D}}_i$ and $\widetilde{\mathcal {D}}_i$.
\begin{definition}
Let $G_i$ be the generator matrix for the quadratic residue codes $\mathcal {E}_i$. Then we define $\widehat{\mathcal {D}}_i$ and $\widetilde{\mathcal {D}}_i$ with $\widehat{G}_i$ and $\widetilde{G}_i$ as their generator matrices as follows, respectively.\\
   {\rm (i)}~If $p\equiv -1({\rm mod}8)$, then
$$\widehat{G}_i=\left[\begin{array}{cccc}3&3&\cdots&3\\ 0& & &\\\vdots& & \huge{G}_i &\\0& & & \end{array}\right]~~and~~\widetilde{G}_i=\left[\begin{array}{cccc}1&3&\cdots&3\\ 0& & &\\\vdots& & \huge{G}_i &\\0& & & \end{array}\right].$$
{\rm(ii)}~If $p\equiv 1({\rm mod}8)$, then
$$\widehat{G}_i=\left[\begin{array}{cccc}3&1&\cdots&1\\ 0& & &\\\vdots& & \huge{G}_i &\\0& & & \end{array}\right]~~and~~\widetilde{G}_i=\left[\begin{array}{cccc}1&1&\cdots&1\\ 0& & &\\\vdots& & \huge{G}_i &\\0& & & \end{array}\right].$$
\end{definition}
\begin{theorem}
Let $\mathcal {D}_i$ be the quadratic residue codes of length $p$ over $R$. The following hold\\
{\rm (i)}~If $p\equiv -1({\rm mod}8)$, then $\widehat{\mathcal {D}}_i$ and $\widetilde{\mathcal {D}}_i$ are Euclidean self-dual. \\
{\rm (ii)}~If $p\equiv 1({\rm mod}8)$, then $\widehat{\mathcal {D}}_1^\perp=\widetilde{\mathcal {D}}_2$ and $\widehat{\mathcal {D}}_2^\perp=\widetilde{\mathcal {D}}_1$.
\end{theorem}
\begin{proof}
If $p\equiv -1({\rm mod}8)$, by the fact that the sum of the components of any codeword in $\mathcal {E}_i$ is zero, we have $\widehat{\mathcal {D}}_i$ and $\widetilde{\mathcal {D}}_i$ are Euclidean self-orthogonal. Furthermore, $|\mathcal {D}_i|=|\widehat{\mathcal {D}}_i|=|\widetilde{\mathcal {D}}_i|=4^{p+1}$ implying $\widehat{\mathcal {D}}_i$ and $\widetilde{\mathcal {D}}_i$ are Euclidean self-dual.
\vskip 2mm
If $p\equiv 1({\rm mod}8)$, then $\mathcal {E}_1^\perp=\mathcal {D}_2$ and $\mathcal {E}_2^\perp=\mathcal {D}_1$. Hence the extended codewords arising from $\mathcal {E}_i$ are orthogonal to all codewords in either $\widehat{\mathcal {D}}_j$ and $\widetilde{\mathcal {D}}_j$ where $j\neq i$. Since the product of the vectors $(3,1,\ldots, 1)$ and $(1,1,\ldots,1)$ is $3+p\equiv 0({\rm mod}4)$, we have $\widehat{\mathcal {D}}_j^\perp\subseteq \widetilde{\mathcal {D}}_i$ where $j\neq i$. Furthermore, $|\mathcal {D}_i|=|\widehat{\mathcal {D}}_i|=|\widetilde{\mathcal {D}}_i|=4^{p+1}$ implying $\widehat{\mathcal {D}}_j^\perp=\widetilde{\mathcal {D}}_i$ where $j\neq i$.
\end{proof}
\section{Examples}
\begin{example}
In this example, we illustrate some $\mathbb{Z}_4$-Euclidean isodual codes obtained by the construction methods A, B, C and the Gray map $\Phi$.
\par
(i)~Consider the matrix $$G=\left[\begin{array}{cc|cc}1&0&2+v&2\\ 0&1&2&2+v \end{array}\right].$$ Let $\mathcal {C}$ be a linear code generated by $G$ over $R$. Then, by Construction A, we see that $\mathcal {C}$ is a Euclidean isodual code of length $4$ over $R$. As a $\mathbb{Z}_4$-module, $\mathcal {C}$ is generated by
$$\mathbb{G}=\left[\begin{array}{cccc}v&0&3v&2v\\ 0&v&2v&3v \\ 1-v & 0 & 2(1-v) & 2(1-v)\\ 0 & 1-v & 2(1-v) & 2(1-v)\end{array}\right],$$
which implies that
$$\Phi(\mathbb{G})=\left[\begin{array}{cccccccc}0&1&0&0&0&3&0&2\\ 0&0&0&1&0&2&0&3 \\ 1&0&0&0&2&0&2&0\\ 0 & 0 & 1 & 0&2&0&2&0\end{array}\right].$$ The linear code $\Phi(\mathcal {C})$ generated by $\Phi(\mathbb{G})$ is a $\mathbb{Z}_4$-Euclidean isodual code of length $8$ with type $4^4$. The Lee weight distribution, Euclidean weight distribution and Hamming weight distribution of $\Phi(\mathcal {C})$ are given as follows, reslectively.
$$W_L(y)=1+6y^2+15y^4+4y^5+84y^6+4y^7+15y^8+\cdots.$$ $$W_E(y)=1+4y^2+6y^4+24y^6+43y^8+\cdots.$$ $$W_H(y)=1+2y+7y^2+16y^3+35y^4+58y^5+65y^6+52y^7+20y^8.$$
\par
(ii)~Consider the matrix $$G=\left[\begin{array}{ccc|ccc}1&0&0&2+v&1+v&1\\ 0&1&0&1&2+v&1+v\\0&0&1&1+v&1&2+v \end{array}\right].$$ Let $\mathcal {C}$ be a linear code generated by $G$ over $R$. Then, by Construction B, we see that $\mathcal {C}$ is a Euclidean isodual code of length $6$ over $R$. The linear code $\Phi(\mathcal {C})$ is a $\mathbb{Z}_4$-Euclidean isodual code length $12$ with type $4^6$. The Lee weight distribution, Euclidean weight distribution and Hamming weight distribution of $\Phi(\mathcal {C})$ are given as follows, reslectively.
$$W_L(y)=1+2y^3+12y^4+42y^5+32y^6+18y^7+102y^8+\cdots.$$ $$W_E(y)=1+2y^3+12y^4+54y^7+60y^8+\cdots.$$ $$W_H(y)=1+10y^3+60y^4+30y^5+50y^6+306y^7+1035y^8+\cdots.$$
\par
(iii)~Consider the matrix $$G=\left[\begin{array}{cccc|cccc}1&0&0&0&2+v&2&2&2\\ 0&1&0&0&2&2+v&1+v&1\\0&0&1&0&2&1&2+v&1+v\\0&0&0&1&2&1+v&1&2+v \end{array}\right].$$ Let $\mathcal {C}$ be a linear code generated by $G$ over $R$. Then, by Construction C, we see that $\mathcal {C}$ is a Euclidean isodual code of length $8$ over $R$. The linear code $\Phi(\mathcal {C})$ is a $\mathbb{Z}_4$-Euclidean isodual code of length $16$ with type $4^8$. The Lee weight distribution, Euclidean weight distribution and Hamming weight distribution of $\Phi(\mathcal {C})$ are given as follows, reslectively.
$$W_L(y)=1+y^2+25y^4+18y^5+75y^6+102y^7+268y^8+\cdots.$$ $$W_E(y)=1+25y^4+16y^5+12y^6+2y^7+157y^8+\cdots.$$ $$W_H(y)=1+y+y^2+9y^3+52y^4+168y^5+254y^6+426y^7+1321y^8+\cdots.$$
\end{example}

 \begin{example}
 In this example, we consider the Euclidean cyclic self-dual codes of length $n\leq 39$ over $R$. By Theorem 14, we have $n=7, 15, 21, 23, 31, 35,$ and $39$.
   \par
   (i)~$n=7$. It is well known that $$X^7-1=(X-1)f(X)(3f^*(X)),$$ where $f(X)=X^3+3X^2+2X+3$. There is only one non-trivial Euclidean cyclic self-dual code over $R$. It is
   \begin{equation*}
   \mathcal {C}=((X-1)f(X), 2f(X)f^*(X)).
   \end{equation*}
   By Theorem 2, the Gray image $\Phi(\mathcal {C})$ is a Euclidean self-dual code of length $14$ with type $4^62^2$ over $\mathbb{Z}_4$. Moreover, $\Phi(\mathcal {C})$ is with minimum Lee distance $4$, i.e., $\Phi(\mathcal {C})$ is a quaternary $(n, M, d_L)=(14, 4^62^2, 4)$ Type I code.
   \par
   (ii)~$n=15$. It is well known that $$X^{15}-1=(X-1)(X^4+X^3+X^2+X+1)(X^2+X+1)f(X)f^*(X),$$ where $f(X)=X^4+2X^2+3X+1$. There is only one non-trivial Euclidean cyclic self-dual code of length $15$ over $R$. It is
   \begin{equation*}
   \mathcal {C}=(f(X)h(X), 2f(X)g(X)),
   \end{equation*}
   where $h(X)=(X-1)(X^4+X^3+X^2+X+1)(X^2+X+1)$ and $g(X)=X^4+3X^3+2X^2+1$. By Theorem 2, the Gray image $\Phi(\mathcal {C})$ is a Euclidean self-dual code of length $30$ with type $4^{8}2^{14}$ over $\mathbb{Z}_4$. Moreover, $\Phi(\mathcal {C})$ is with minimum Lee distance $6$, i.e., $\Phi(\mathcal {C})$ is a quaternary $(n, M, d_L)=(30, 4^82^{14}, 6)$ Type I code.
   \par
   (iii)~$n=21.$ It is well known that $$X^{21}-1=(X-1)(X^2+X+1)f_1(X)f_1^*(X)f_2(X)(3f_2^*(X)),$$ where $f_1(X)=X^6+2X^5+3X^4+3X^2+X+1$, $f_2(X)=X^3+2X^2+X+3$, $h_1(X)=X^9+X^8+X^7+3X^2+3X+3$, $h_2(X)=X^{15}+3X^{14}+X^8+3X^7+X+3$ and $h_3(X)=X^3+3$. There are $9$ different non-trivial Euclidean cyclic self-dual codes of length $21$ over $R$. We illustrate them in Table 1.
   \begin{table}
   \caption{Euclidean cyclic self-dual codes of length $21$ over $R$}
   \begin{center}
   \begin{small}
   \begin{tabular}{ccc}
   \hline
   Codes number &  Generators of cyclic self-dual codes &  Gray images \\
   \hline
   $\mathcal {C}_1$ & \{ $f_1h_1, 2f_1f^*_1$\} & $(42, 4^{12}2^{18}, 6)$\\
   $\mathcal {C}_2$ & \{ $vf_1h_1+(1-v)f_2h_2, 2vf_1f^*_1+2(1-v)f_2f^*_2$\} & $(42, 4^{9}2^{24}, 4)$\\
   $\mathcal {C}_3$ & \{ $vf_1h_1+(1-v)f_1f_2h_3, 2vf_1f^*_1+2(1-v)f_1f_2f^*_1f^*_2$\} & $(42, 4^{15}2^{12}, 4)$\\
   $\mathcal {C}_4$ & \{ $f_2h_2, 2f_2f^*_2$\} & $(42, 4^{6}2^{30}, 4)$\\
   $\mathcal {C}_5$ & \{ $vf_2h_2+(1-v)f_1h_1, 2vf_2f^*_2+2(1-v)f_1f^*_1$\} & $(42, 4^{9}2^{24}, 4)$\\
   $\mathcal {C}_6$ & \{ $vf_2h_2+(1-v)f_1f_2h_3, 2vf_2f^*_2+2(1-v)f_1f_2f^*_1f^*_2$\} & $(42, 4^{12}2^{18}, 4)$\\
   $\mathcal {C}_7$ & \{ $f_1f_2h_3, 2f_1f_2f^*_1f^*_2$\} & $(42, 4^{18}2^6, 4)$\\
   $\mathcal {C}_8$ & \{ $vf_1f_2h_3+(1-v)f_1h_1, 2vf_1f_2f^*_1f^*_2+2(1-v)f_1f^*_1$\} & $(42, 4^{15}2^{12}, 4)$\\
   $\mathcal {C}_9$ & \{ $vf_1f_2h_3+(1-v)f_2h_2, 2vf_1f_2f^*_1f^*_2+2(1-v)f_2f^*_2$\} & $(42, 4^{12}2^{24}, 4)$\\
   \hline
   \end{tabular}
   \end{small}
   \end{center}
   \end{table}
\par
(iv)~$n=23.$ It is well known that $$X^{23}-1=(X-1)f(X)(3f^*(X)),$$ where $f(X)=X^{11}+2X^{10}+3X^9+3X^7+3X^6+3X^5+2X^4+X+3$. There is only one non-trivial Euclidean cyclic self-dual code of length $23$ over $R$. It is
\begin{equation*}
\mathcal {C}=((X-1)f(X), 2f(X)f^*(X)).
\end{equation*}
By Theorem 2, the Gray image $\Phi(\mathcal {C})$ is a Euclidean self-dual code of length $46$ with type $4^{22}2^2$ over $\mathbb{Z}_4$. Moreover, $\Phi(\mathcal {C})$ is with minimum Lee distance $7$, i.e., $\Phi(\mathcal {C})$ is a quaternary $(46, 4^{22}2^2, 7)$ Type I code.
\par
(v)~$n=31.$ It is well known that $$X^{31}-1=(X-1)f_1(X)(3f^*_1(X))f_2(X)(3^*_2(X))f_3(X)(3f^*_3(X)),$$ where $f_1(X)=X^5+3X^2+2X+3$, $f_2(X)=X^5+2X^4+3X^3+X^2+3X+3$ and $f_3(X)=X^5+3X^4+X^2+3X+3$. Let $h_1(X)=(X-1)f_2(X)f^*_2(X)f_3(X)f^*_3(X)$, $h_2(X)=h_3(X)=(X-1)f_3(X)f^*_3(X)$ and $h_4(X)=h_5(X)=X-1$. There are $25$ different non-trivial Euclidean cyclic self-dual codes of length $31$ over $R$. We illustrate them in Table 2.
    \begin{table}
   \caption{Euclidean cyclic self-dual codes of length $31$ over $R$}
   \begin{center}
   \begin{small}
   \begin{tabular}{ccc}
   \hline
   Codes number &  Generators of cyclic self-dual codes &  Gray images \\
   \hline
   $\mathcal {C}_1$ & \{ $f_1h_1, 2f_1f^*_1$\} & $(62, 4^{10}2^{42}, 6)$\\
   $\mathcal {C}_2$ & \{ $vf_1h_1+(1-v)f_1f_2h_2, 2vf_1f^*_1+2(1-v)f_1f_2f^*_1f^*_2$\} & $(62, 4^{15}2^{32}, 12)$\\
   $\mathcal {C}_3$ & \{ $vf_1h_1+(1-v)f_1f^*_2h_3, 2vf_1f^*_1+2(1-v)f_1f_2f^*_1f^*_2$\} & $(62, 4^{15}2^{32}, 12)$\\
   $\mathcal {C}_4$ & \{ $vf_1h_1+(1-v)f_1f_2f_3h_4, 2vf_1f^*_1+2(1-v)f_1f_2f_3f^*_1f^*_2f^*_3$\} & $(62, 4^{20}2^{22},6)$\\
   $\mathcal {C}_5$ & \{ $vf_1h_1+(1-v)f_1f^*_2f_3h_5, 2vf_1f^*_1+2(1-v)f_1f_2f_3f^*_1f^*_2f^*_3$\} & $(62, 4^{20}2^{22},6)$\\
   $\mathcal {C}_6$ & \{ $f_1f_2h_2, 2f_1f_2f^*_1f^*_2$\} & $(62, 4^{20}2^{22}, 10)$\\
   $\mathcal {C}_7$ & \{ $vf_1f_2h_2+(1-v)f_1h_1, 2vf_1f_2f^*_1f^*_2+2(1-v)f_1f^*_1$\} & $(62, 4^{15}2^{32}, 12)$\\
   $\mathcal {C}_8$ & \{ $vf_1f_2h_2+(1-v)f_1f^*_2h_3, 2f_1f_2f^*_1f^*_2$\} & $(62, 4^{20}2^{22}, 10)$\\
   $\mathcal {C}_9$ & \{ $vf_1f_2h_2+(1-v)f_1f_2f_3h_4, 2vf_1f_2f^*_1f^*_2+2(1-v)f_1f_2f_3f^*_1f^*_2f^*_3$\} & $(62, 4^{25}2^{12}, 10)$\\
   $\mathcal {C}_{10}$ & \{ $vf_1f_2h_2+(1-v)f_1f^*_2f_3h_5, 2vf_1f_2f^*_1f^*_2+2(1-v)f_1f_2f_3f^*_1f^*_2f^*_3$\} & $(62, 4^{25}2^{12}, 10)$\\
   $\mathcal {C}_{11}$ & \{ $f_1f^*_2h_3, 2f_1f_2f^*_1f^*_2$\} & $(62, 4^{20}2^{22}, 10)$\\
   $\mathcal {C}_{12}$ & \{ $vf_1f^*_2h_3+(1-v)f_1h_1, 2vf_1f_2f^*_1f^*_2+2(1-v)f_1f^*_1$\} & $(62, 4^{15}2^{32}, 14)$\\
   $\mathcal {C}_{13}$ & \{ $vf_1f^*_2h_3+(1-v)f_1f_2h_2, 2f_1f_2f^*_1f^*_2$\} & $(62, 4^{20}2^{22}, 10)$\\
   $\mathcal {C}_{14}$ & \{ $vf_1f^*_2h_3+(1-v)f_1f_2f_3h_4, 2vf_1f_2f^*_1f^*_2+2(1-v)f_1f_2f_3f^*_1f^*_2f^*_3$\} & $(62,4^{25}2^{12}, 10)$\\
   $\mathcal {C}_{15}$ & \{ $vf_1f^*_2h_3+(1-v)f_1f^*_2f_3h_5, 2vf_1f_2f^*_1f^*_2+2(1-v)f_1f_2f_3f^*_1f^*_2f^*_3$\} & $(62, 4^{25}2^{12}, 10)$\\
   $\mathcal {C}_{16}$ & \{ $f_1f_2f_3h_4, 2f_1f_2f_3f^*_1f^*_2f^*_3$\} & $(62, 4^{30}2^2, 12)$\\
   $\mathcal {C}_{17}$ & \{ $vf_1f_2f_3h_4+(1-v)f_1h_1, 2vf_1f_2f_3f^*_1f^*_2f^*_3+2(1-v)f_1f^*_1$\} & $(62, 4^{20}2^{22}, 6)$\\
   $\mathcal {C}_{18}$ & \{ $vf_1f_2f_3h_4+(1-v)f_1f_2h_2, 2vf_1f_2f_3f^*_1f^*_2f^*_3+2(1-v)f_1f_2f^*_1f^*_2$\} & $(62, 4^{25}2^{12}, 10)$\\
   $\mathcal {C}_{19}$ & \{ $vf_1f_2f_3h_4+(1-v)f_1f^*_2h_3, 2vf_1f_2f_3f^*_1f^*_2f^*_3+2(1-v)f_1f_2f^*_1f^*_2$\} & $(62, 4^{25}2^{12}, 10)$\\
   $\mathcal {C}_{20}$ & \{ $vf_1f_2f_3h_4+(1-v)f_1f^*_2f_3h_5, 2vf_1f_2f_3f^*_1f^*_2f^*_3+2(1-v)f_1f_2f_3f^*_1f^*_2f^*_3$\} & $(62, 4^{30}2^2, 12)$\\
   $\mathcal {C}_{21}$ & \{ $f_1f^*_2f_3h_5, 2f_1f_2f_3f^*_1f^*_2f^*_3$\} & $(62, 4^{30}2^2, 12)$\\
   $\mathcal {C}_{22}$ & \{ $vf_1f^*_2f_3h_5+(1-v)f_1h_1, 2vf_1f_2f_3f^*_1f^*_2f^*_3+2(1-v)f_1f^*_1$\} & $(62, 4^{20}2^{22}, 6)$\\
   $\mathcal {C}_{23}$ & \{ $vf_1f^*_2f_3h_5+(1-v)f_1f_2h_2, 2vf_1f_2f_3f^*_1f^*_2f^*_3+2(1-v)f_1f_2f^*_1f^*_2$\} & $(62, 4^{25}2^{12}, 10)$\\
   $\mathcal {C}_{24}$ & \{ $vf_1f^*_2f_3h_5+(1-v)f_1f^*_2h_3, 2vf_1f_2f_3f^*_1f^*_2f^*_3+2(1-v)f_1f_2f^*_1f^*_2$\} & $(62, 4^{25}2^{12}, 10)$\\
   $\mathcal {C}_{25}$ & \{ $vf_1f^*_2f_3h_5+(1-v)f_1f_2f_3h_4, 2f_1f_2f_3f^*_1f^*_2f^*_3$\} & $(62, 4^{30}2^2, 12)$\\
   \hline
   \end{tabular}
   \end{small}
   \end{center}
   \end{table}
\par
   (vi)~$n=35.$ It is well known that $$X^{35}-1=f_1(X)f^*_1(X)f_{2}(X)f^*_{2}(X)h(X),$$ where $f_1(X)=X^3+2X^2+X+3$, $f_2(X)=X^{12}+2X^{11}+
   3X^{10}+X^9+X^8+3X^7+2X^6+2X^5+X^4+2X^3+3X^2+X+1$ and $h(X)=(X-1)(X^4+X^3+X^2+X+1)$. There are $16$ different non-trivial Euclidean cyclic self-dual codes over $R$. We illustrate them in Table 3.
   \begin{table}
   \caption{Euclidean cyclic self-dual codes of length $35$ over $R$}
   \begin{center}
   \begin{small}
   \begin{tabular}{ccc}
   \hline
   Codes number &  Generators of cyclic self-dual codes &  Gray images \\
   \hline
   $\mathcal {C}_1$ & \{ $f_1f_2h, 2f_1f_2f^*_1f^*_2$\} & $(70, 4^{30}2^{10}, 4)$\\
   $\mathcal {C}_2$ & \{ $vf_1f_2h+(1-v)f^*_1f_2h, 2f_1f_2f^*_1f^*_2$\} & $(70, 4^{30}2^{10}, 4)$\\
   $\mathcal {C}_3$ & \{ $vf_1f_2h+(1-v)f_1f^*_1f_2h, 2vf_1f_2f^*_1f^*_2+2(1-v)f_2f^*_2$\} & $(70, 4^{27}2^{16}, 4)$\\
   $\mathcal {C}_4$ & \{ $vf_1f_2h+(1-v)f_1f^*_2f_2h, 2vf_1f_2f^*_1f^*_2+2(1-v)f_1f^*_1$\} & $(70, 4^{18}2^{34}, 4)$\\
   $\mathcal {C}_5$ & \{ $f^*_1f_2h, 2f_1f_2f^*_1f^*_2$\} & $(70, 4^{30}2^{10}, 8)$\\
   $\mathcal {C}_6$ & \{ $vf^*_1f_2h+(1-v)f_1f_2h, 2f_1f_2f^*_1f^*_2$\} & $(70, 4^{30}2^{10}, 4)$\\
   $\mathcal {C}_7$ & \{ $vf^*_1f_2h+(1-v)f_1f^*_1f_2h, 2vf_1f_2f^*_1f^*_2+2(1-v)f_2f^*_2$\} & $(70, 4^{27}2^{16}, 6)$\\
   $\mathcal {C}_8$ & \{ $vf^*_1f_2h+(1-v)f_1f_2f^*_2h, 2vf_1f_2f^*_1f^*_2+2(1-v)f_1f^*_1$\} & $(70, 4^{18}2^{34}, 4)$\\
   $\mathcal {C}_9$ & \{ $f_1f^*_1f_2h, 2f_2f^*_2$\} & $(70, 4^{24}2^{22}, 6)$\\
   $\mathcal {C}_{10}$ & \{ $vf_1f^*_1f_2h+(1-v)f_1f_2h, 2vf_2f^*_2+2(1-v)f_1f_2f^*_1f^*_2$\} & $(70, 4^{27}2^{16}, 4)$\\
   $\mathcal {C}_{11}$ & \{ $vf_1f^*_1f_2h+(1-v)f^*_1f_2h, 2vf_2f^*_2+2(1-v)f_1f_2f^*_1f^*_2$\} & $(70, 4^{27}2^{16}, 6)$\\
   $\mathcal {C}_{12}$ & \{ $vf_1f^*_1f_2h+(1-v)f_1f_2f^*_2h, 2vf_2f^*_2+2(1-v)f_1f^*_1$\} & $(70, 4^{15}2^{40}, 4)$\\
   $\mathcal {C}_{13}$ & \{ $f_1f_2f^*_2h, 2f_1f^*_1$\} & $(70, 4^{6}2^{58}, 6)$\\
   $\mathcal {C}_{14}$ & \{ $vf_1f_2f^*_2h+(1-v)f_1f_2h, 2vf_1f^*_1+2(1-v)f_1f_2f^*_1f^*_2$\} & $(70, 4^{18}2^{34}, 4)$\\
   $\mathcal {C}_{15}$ & \{ $vf_1f_2f^*_2h+(1-v)f^*_1f_2h, 2vf_1f^*_1+2(1-v)f_1f_2f^*_1f^*_2$\} & $(70, 4^{18}2^{34}, 4)$\\
   $\mathcal {C}_{16}$ & \{ $vf_1f_2f^*_2h+(1-v)f_1f^*_1f_2h, 2vf_1f^*_1+2(1-v)f_2f^*_2$\} & $(70, 4^{15}2^{40}, 4)$\\
   \hline
   \end{tabular}
   \end{small}
   \end{center}
   \end{table}
   \par
   (vii)~$n=39.$ It is well known that $$X^{39}-1=f(X)f^*(X)h(X),$$ where $f(X)=X^{12}+X^{11}-X^{10}-X^9+2X^6+X^5-X^4+X^3-X^2+2X+1$ and $h(X)=(X-1)(X^2+X+1)(X^{12}+X^{11}+\cdots+X+1)$. There is only one non-trivial Euclidean cyclic self-dual code of length $39$ over $R$. It is
   \begin{equation*}
   \mathcal {C}=(f(X)h(X), 2f(X)f^*(X)).
   \end{equation*}
   By Theorem 2, the Gray image $\Phi(\mathcal {C})$ is a Euclidean self-dual code of length $78$ with type $4^{24}2^{30}$ over $\mathbb{Z}_4$. Moreover, $\Phi(\mathcal {C})$ is with minimum Lee distance 6, i.e., $\Phi(\mathcal {C})$ is a quaternary $(78, 4^{24}2^{30}, 6)$ Type I code.
   \end{example}
\begin{example}
   In this example, compared to the linear codes in table of The $\mathbb{Z}_4$ Database \cite{Aydin}, we show that some new linear codes over $\mathbb{Z}_4$ with improved minimum Lee weight can be constructed from the cyclic codes over $R$. We do not list the generator matrices of these linear codes here for the interest of space. If needed, they are available from the authors.
   \par
   (i)~It is well known that $$X^{23}-1=(X-1)f(X)(3f^*(X)),$$ where $f(X)=X^{11}+2X^{10}+3X^9+3X^7+3X^6+3X^5+2X^4+X+3$. Let
\begin{equation*}
\mathcal {C}=((X-1)f(X)).
\end{equation*}
Then $\Phi(\mathcal {C})$ is a $\mathbb{Z}_4$-linear $(46, 4^{22})$ code with minimum Lee weight $8$, which is better than $6$.
\par
(ii)~It is well known that $$X^{31}-1=(X-1)f_1(X)(3f^*_1(X))f_2(X)(3^*_2(X))f_3(X)(3f^*_3(X)),$$ where $f_1(X)=X^5+3X^2+2X+3$, $f_2(X)=X^5+2X^4+3X^3+X^2+3X+3$ and $f_3(X)=X^5+3X^4+X^2+3X+3$. Let $h_1(X)=(X-1)f_2(X)f^*_2(X)f_3(X)f^*_3(X)$, $h_2(X)=h_3(X)=(X-1)f_3(X)f^*_3(X)$ and $h_4(X)=h_5(X)=X-1$. We list $21$ new $\mathbb{Z}_4$-linear codes of length $62$ from the cyclic codes of length $31$ over $R$ in Table $4$.
\begin{table}
   \caption{The $21$ new $\mathbb{Z}_4$-linear codes of length $62$}
   \begin{center}
   \begin{small}
   \begin{tabular}{ccc}
   \hline
   Codes number &  Generators of cyclic self-dual codes &  Gray images \\
   \hline
   $\mathcal {C}_1$ & \{ $vf_1h_1+(1-v)f_1f_2h_2$\} & $(62, 4^{15}, 16)$\\
   $\mathcal {C}_2$ & \{ $vf_1h_1+(1-v)f_1f_2f_3h_4$\} & $(62, 4^{20},14)$\\
   $\mathcal {C}_3$ & \{ $(1-v)f_1f^*_2h_3$\} & $(62, 4^{15},16)$\\
   $\mathcal {C}_4$ & \{ $f_1h_1$\} & $(62, 4^{10},28)$\\
   $\mathcal {C}_5$ & \{ $vf_1h_1+(1-v)f_1f^*_2f_3h_5$\} & $(62, 4^{20},14)$\\
   $\mathcal {C}_6$ & \{ $f_1f_2h_2$\} & $(62, 4^{20}, 16)$\\
   $\mathcal {C}_7$ & \{ $vf_1f_2h_2+(1-v)f_1h_1$\} & $(62, 4^{15}, 16)$\\
   $\mathcal {C}_8$ & \{ $vf_1f_2h_2+(1-v)f_1f^*_2h_3$\} & $((62, 4^{20}, 16)$\\
   $\mathcal {C}_9$ & \{ $vf_1f_2h_2+(1-v)f_1f_2f_3h_4$\} & $(62, 4^{25}, 14)$\\
   $\mathcal {C}_{10}$ & \{ $vf_1f_2h_2+(1-v)f_1f^*_2f_3h_5$\} & $(62, 4^{25}, 14)$\\
   $\mathcal {C}_{11}$ & \{ $f_1f^*_2h_3$\} & $(62, 4^{20}, 18)$\\
   $\mathcal {C}_{12}$ & \{ $vf_1f^*_2h_3+(1-v)f_1h_1$\} & $(62, 4^{15}, 18)$\\
   $\mathcal {C}_{13}$ & \{ $vf_1f^*_2h_3+(1-v)f_1f_2h_2$\} & $(62, 4^{20}, 16)$\\
   $\mathcal {C}_{14}$ & \{ $vf_1f^*_2h_3+(1-v)f_1f_2f_3h_4$\} & $(62,4^{25}, 14)$\\
   $\mathcal {C}_{15}$ & \{ $vf_1f^*_2h_3+(1-v)f_1f^*_2f_3h_5$\} & $(62, 4^{25}, 14)$\\
   $\mathcal {C}_{16}$ & \{ $f_1f_2f_3h_4$\} & $(62, 4^{30}, 14)$\\
   $\mathcal {C}_{17}$ & \{ $vf_1f_2f_3h_4+(1-v)f_1f_2h_2$\} & $(62, 4^{25}, 14)$\\
   $\mathcal {C}_{18}$ & \{ $vf_1f_2f_3h_4+(1-v)f_1f^*_2h_3$\} & $(62, 4^{25}, 14)$\\
   $\mathcal {C}_{19}$ & \{ $vf_1f^*_2f_3h_5+(1-v)f_1h_1$\} & $(62, 4^{20}, 14)$\\
   $\mathcal {C}_{20}$ & \{ $vf_1f^*_2f_3h_5+(1-v)f_1f_2h_2$\} & $(62, 4^{25}, 14)$\\
   $\mathcal {C}_{21}$ & \{ $vf_1f^*_2f_3h_5+(1-v)f_1f_2f_3h_4$\} & $(62, 4^{30}, 14)$\\
   \hline
   \end{tabular}
   \end{small}
   \end{center}
   \end{table}
\par
(iii)~ It is well known that $$X^{35}-1=f_1(X)f^*_1(X)f_{2}(X)f^*_{2}(X)h(X),$$ where $f_1(X)=X^3+2X^2+X+3$, $f_2(X)=X^{12}+2X^{11}+
   3X^{10}+X^9+X^8+3X^7+2X^6+2X^5+X^4+2X^3+3X^2+X+1$ and $h(X)=(X-1)(X^4+X^3+X^2+X+1)$. We list $4$ new $\mathbb{Z}_4$-linear codes of length $70$ from the cyclic codes of length $35$ over $R$ in Table $5$.
\begin{table}
   \caption{The $4$ new $\mathbb{Z}_4$-linear codes of length $70$}
   \begin{center}
   \begin{small}
   \begin{tabular}{ccc}
   \hline
   Codes number &  Generators of cyclic self-dual codes &  Gray images \\
   \hline
   $\mathcal {C}_1$ & \{ $f^*_1f_2h$\} & $(70, 4^{30}, 12)$\\
   $\mathcal {C}_2$ & \{ $vf^*_1f_2h+(1-v)f_1f^*_1f_2h$\} & $(70, 4^{27}, 12)$\\
   $\mathcal {C}_3$ & \{ $f_1f^*_1f_2h$\} & $(70, 4^{24}, 12)$\\
   $\mathcal {C}_{4}$ & \{ $vf_1f^*_1f_2h+(1-v)f^*_1f_2h$\} & $(70, 4^{27}, 12)$\\
   \hline
   \end{tabular}
   \end{small}
   \end{center}
   \end{table}
\par
(iv)~It is well known that $$X^{39}-1=f(X)f^*(X)h(X),$$ where $f(X)=X^{12}+X^{11}-X^{10}-X^9+2X^6+X^5-X^4+X^3-X^2+2X+1$ and $h(X)=(X-1)(X^2+X+1)(X^{12}+X^{11}+\cdots+X+1)$. Let
\begin{equation*}
\mathcal {C}=(f(X)h(X)).
\end{equation*}
Then $\Phi(\mathcal {C})$ is a $\mathbb{Z}_4$-linear $(78, 4^{24})$ code with minimum Lee weight $16$, which is better than $10$.
\end{example}
\begin{example}
 Let $p=7$. We consider the quadratic residue codes of length $7$ over $R$. By the definitions of $Q(X)$ and $N(X)$, we have $Q(X)=X+X^2+X^4$ and $N(X)=X^3+X^5+X^6$. Since $7\equiv  -1({\rm mod}8)$, by Definition $5$, it follows that
\begin{equation*}
\mathcal {D}_1=(v(X+X^2+2X^3+X^4+2X^5+2X^6)+(1-v)(2X+2X^2+X^3+2X^4+X^5+X^6)),
\end{equation*}
\begin{equation*}
\mathcal {D}_2=(v(2X+2X^2+X^3+2X^4+X^5+X^6)+(1-v)(X+X^2+2X^3+X^4+2X^5+2X^6)),
\end{equation*}
\begin{equation*}
\begin{split}
\mathcal {E}_1&= (v(1+2X+2X^2+3X^3+2X^4+3X^5+3X^6)\\
              &+(1-v)(1+3X+3X^2+2X^3+3X^4+2X^5+2X^6))
 \end{split}
 \end{equation*}
and
\begin{equation*}
\begin{split}
\mathcal {E}_2&=(v(1+3X+3X^2+2X^3+3X^4+2X^5+2X^6)\\
              &+(1-v)(1+2X+2X^2+3X^3+2X^4+3X^5+3X^6))
 \end{split}
 \end{equation*}
are quadratic residue codes of length $7$ over $R$. By Theorem $16$, $\mathcal {E}_1$ and $\mathcal {E}_2$ can be regarded as the $\mathbb{Z}_4[X]$-modules, i.e. \begin{equation*}
\begin{split}
\mathcal {E}_1&= v(1+2X+2X^2+3X^3+2X^4+3X^5+3X^6) \\
              &\oplus (1-v)(1+3X+3X^2+2X^3+3X^4+2X^5+2X^6)
 \end{split}
 \end{equation*}
and
\begin{equation*}
\begin{split}
\mathcal {E}_2&= v(1+3X+3X^2+2X^3+3X^4+2X^5+2X^6) \\
              &\oplus (1-v)(1+2X+2X^2+3X^3+2X^4+3X^5+3X^6),
 \end{split}
\end{equation*}
which implies that $\mathcal {E}_1$ and $\mathcal {E}_2$ have the following $\mathbb{Z}_4$-generator matrices
$$G_1=\left[\begin{array}{c}vG_{11}\\ (1-v)G_{12} \end{array}\right]~~{\rm and}~~G_2=\left[\begin{array}{c}vG_{21}\\ (1-v)G_{22} \end{array}\right],$$
where
$$G_{11}=G_{22}=\left[\begin{array}{ccccccc} 1&0&0&1&2&3&1 \\ 0&1&0&3&3&3&2 \\ 0&0&1&2&3&1&1\end{array}\right],$$
$$G_{12}=G_{21}=\left[\begin{array}{ccccccc} 1&0&0&1&1&3&2 \\ 0&1&0&2&3&3&3 \\ 0&0&1&1&3&2&1\end{array}\right].$$
By Definition 7, we have $\widehat{\mathcal {D}}_1$, $\widetilde{\mathcal {D}}_1$ and $\widehat{\mathcal {D}}_2$, $\widetilde{\mathcal {D}}_2$ are the extensions of $\mathcal {D}_1$ and $\mathcal {D}_2$, respectively. Furthermore, thay have the generator matrices as follows
$$\widehat{G}_1=\left[\begin{array}{cccc}3&3&\cdots&3\\ 0& & &\\\vdots& & \huge{G}_1 &\\0& & & \end{array}\right]~~,~~\widetilde{G}_1=\left[\begin{array}{cccc}1&3&\cdots&3\\ 0& & &\\\vdots& & \huge{G}_1 &\\0& & & \end{array}\right],$$
$$\widehat{G}_2=\left[\begin{array}{cccc}3&3&\cdots&3\\ 0& & &\\\vdots& & \huge{G}_2 &\\0& & & \end{array}\right]~~{\rm and}~~\widetilde{G}_2=\left[\begin{array}{cccc}1&3&\cdots&3\\ 0& & &\\\vdots& & \huge{G}_2 &\\0& & & \end{array}\right].$$
$\widehat{\mathcal {D}}_1$, $\widetilde{\mathcal {D}}_1$, $\widehat{\mathcal {D}}_2$ and $\widetilde{\mathcal {D}}_2$ are equivalent and, by Theorem 20, they are extremal Type II codes. Theorefore, by Theorem 7, the Gray images of $\widehat{\mathcal {D}}_1$, $\widetilde{\mathcal {D}}_1$, $\widehat{\mathcal {D}}_2$ and $\widetilde{\mathcal {D}}_2$ are extremal Type II codes of length $16$ with minimum Euclidean weight $8$ over $\mathbb{Z}_4$. The Euclidean weight distributions of these codes are given as follows
$$W_E(y)=1+256y^8+16636y^{16}+32256y^{24}+15878y^{32}+256y^{40}+252y^{48}+y^{64}.$$
\end{example}

\end{document}